\documentclass{article}

\usepackage[preprint]{neurips_2025}

\usepackage[utf8]{inputenc} % allow utf-8 input
\usepackage[T1]{fontenc}    % use 8-bit T1 fonts
\usepackage{hyperref}       % hyperlinks
\usepackage{url}            % simple URL typesetting
\usepackage{booktabs}       % professional-quality tables
\usepackage{amsfonts}       % blackboard math symbols
\usepackage{nicefrac}       % compact symbols for 1/2, etc.
\usepackage{microtype}      % microtypography
\usepackage{xcolor}         % colors
\usepackage{amsmath,amsthm}
\newcommand{\polylog}{\mbox{\rm polylog}}
\DeclareMathOperator{\poly}{poly}

\usepackage{algorithm,algpseudocode}

\usepackage[capitalize,nameinlink,noabbrev]{cleveref}
\hypersetup{
   colorlinks,
   linkcolor={red!100!black},
   citecolor={green!100!black},
}
\usepackage{enumitem}

\newtheorem{theorem}{Theorem}
\newtheorem{fact}[theorem]{Fact}
\newtheorem{definition}[theorem]{Definition}
\newtheorem{remark}[theorem]{Remark}

\newtheorem{lemma}[theorem]{Lemma}
\newtheorem{claim}[theorem]{Claim}
\newtheorem{result}{Result}
\usepackage{makecell}
\usepackage{multicol}

\title{Do you know what q-means?}

\author{%
  Arjan Cornelissen\\ %\thanks{https://arriopolis.github.io/} \\
  Simons Institute \\
  UC Berkeley\\
  California, USA \\
  \texttt{ajcornelissen@outlook.com} \\
  \And
  Joao F. Doriguello\\ %\thanks{www.joaodoriguello.com} \\
  HUN-REN Alfr\'ed R\'enyi Institute of Mathematics\\
  Budapest, Hungary\\
  \texttt{doriguello@renyi.hu} \\
  \AND
  Alessandro Luongo\\
  Centre for Quantum Technologies\\
  National University of Singapore\\
  Singapore\\
  Inveriant Pte. Ltd.\\
  \texttt{ale@nus.edu.sg} \\
  \And
  Ewin Tang\\
  Department of EECS\\
  UC Berkeley\\
  California, USA \\
  \texttt{ewin@berkeley.edu}
}

\begin{document}

\maketitle

\begin{abstract}
Clustering is one of the most important tools for analysis of large datasets, and perhaps the most popular clustering algorithm is Lloyd's algorithm for $k$-means.
This algorithm takes $n$ vectors $V=[v_1,\dots,v_n]\in\mathbb{R}^{d\times n}$ and outputs $k$ centroids $c_1,\dots,c_k\in\mathbb{R}^d$; these partition the vectors into clusters based on which centroid is closest to a particular vector. We present a classical $\varepsilon$-$k$-means algorithm that performs an approximate version of one iteration of Lloyd's algorithm with time complexity $\widetilde{O}\big(\frac{\|V\|_F^2}{n}\frac{k^{2}d}{\varepsilon^2}(k + \log{n})\big)$, exponentially improving the dependence on the data size $n$ and matching that of the ``$q$-means'' quantum algorithm originally proposed by Kerenidis, Landman, Luongo, and Prakash (NeurIPS'19). Moreover, we propose an improved $q$-means quantum algorithm with time complexity $\widetilde{O}\big(\frac{\|V\|_F}{\sqrt{n}}\frac{k^{3/2}d}{\varepsilon}(\sqrt{k}+\sqrt{d})(\sqrt{k} + \log{n})\big)$ that quadratically improves the runtime of our classical $\varepsilon$-$k$-means algorithm in several parameters.
Our quantum algorithm does not rely on quantum linear algebra primitives of prior work, but instead only uses QRAM to prepare simple states based on the current iteration's clusters and multivariate quantum amplitude estimation. Finally, we provide classical and quantum query lower bounds, showing that our algorithms are optimal in most parameters.
\end{abstract}

\section{Introduction}

Among machine learning problems, data clustering and the $k$-means problem are of particular relevance and have attracted much attention in the past~\cite{hartigan1979algorithm,krishna1999genetic,likas2003global}. Here the task is to find an assignment of each vector from a dataset of size $n$ to one of $k$ labels (for a given $k$ assumed to be known) such that similar vectors are assigned to the same cluster. To be more precise, in the $k$-means problem we are given $n$ vectors $v_1,\dots,v_n\in\mathbb{R}^d$ as columns in a matrix $V\in\mathbb{R}^{d\times n}$, and a positive integer $k$, and the task is to find $k$ centroids $c_1,\dots, c_k\in\mathbb{R}^d$ such that the cost function $\sum_{i\in[n]} \min_{j\in[k]} \|v_i - c_j\|^2$, called \emph{residual sum of squares}, is minimized, where $\|v_i - c_j\|$ is the Euclidean distance between $v_i$ and $c_j$ and $[n] := \{1,\dots,n\}$ for $n\in \mathbb{N}$.

Since the $k$-means problem is known to be NP-hard~\cite{dasgupta2008hardness,vattani2009hardness,mahajan2012planar}, several classical polynomial-time algorithms~\cite{kanungo2002local,jaiswal2014simple,ahmadian2019better,bhattacharya_et_al:LIPIcs:2020:13254} have been developed to obtain \emph{approximate} solutions. One such algorithm is the $k$-means algorithm (also known as Lloyd's algorithm) introduced by Lloyd~\cite{lloyd1982least}, a heuristic algorithm that iteratively updates the centroids $c_1,\dots,c_k$ until some desired precision is reached. At each time step $t$, the algorithm clusters the data into $k$ clusters denoted by the sets $\mathcal{C}_j^t\subseteq[n]$, $j\in[k]$, each with centroid $c_j^t$, and then updates the centroids based on such clustering. More precisely, the $k$-means algorithm starts with initial centroids $c_1^0,\dots,c_k^0\in\mathbb{R}^d$, picked either randomly or through some pre-processing routine as in the $k$-means++ algorithm~\cite{arthur2007k}, and alternates between two steps:
\begin{enumerate}[wide]
    \item Each vector $v_i$, $i\in[n]$, is assigned a cluster $\mathcal{C}_{\ell_i^t}^t$ where $\ell^t_i = \arg\min_{j\in[k]}\|v_i-c_j^t\|$; 
    \item The new centroids $\{c_j^{t+1}\}_{j\in[k]}$ are updated based on the new clusters, $c_j^{t+1} = \frac{1}{|\mathcal{C}_j^t|}\sum_{i\in\mathcal{C}_j^t} v_i$.
\end{enumerate}
The above steps are repeated until $\frac{1}{k}\sum_{j\in[k]} \|c_j^t - c_j^{t+1}\| \leq T$ for a given threshold $T > 0$, in which case we say that the algorithm has converged. The naive runtime of a single iteration is $O(knd)$.

On the other hand, quantum computing is a fast-paced field that promises a new paradigm to solve certain tasks more efficiently~\cite{shor1994algorithms,shor1999polynomial,grover1996fast,grover1997quantum}. The subfield of quantum machine learning aims to offer computational advantages in machine learning, with many new proposed algorithms~\cite{lloyd2014quantum,kerenidis2016quantum,kerenidis2020quantum,chakraborty2018power,lloyd2013quantum,allcock2020quantum}. A notable line of work in this subfield is quantum versions of the $k$-means algorithm. Aimeur, Brassard, and Gambs~\cite{aimeur2013quantum} proposed a quantum algorithm based on quantum minimum finding~\cite{durr1996quantum} for executing a single Lloyd's iteration for the $k$-median clustering problem (which is related to the $k$-means problem), which runs in time superlinear in the input size $n$. Lloyd, Mohseni, and Rebentrost~\cite{lloyd2013quantum} gave a quantum algorithm for approximately executing a single Lloyd's iteration from the $k$-means problem in time linear in $n$ using an efficient subroutine for quantum distance estimation assuming quantum access to the dataset via a QRAM (Quantum Random Access Memory)~\cite{giovannetti2008architectures,giovannetti2008quantum}. 

More recently, Kerenidis, Landman, Luongo, and Prakash~\cite{kerenidis2019q} proposed a quantum version of Lloyd's algorithm called $q$-means. They use quantum linear algebra subroutines and assume the input vectors are stored in QRAM and that all clusters are of roughly equal size. Their per-iteration runtime depends only polylogarithmically on the size $n$ of the dataset, an exponential improvement compared to prior work. However, their quantum algorithm only performs each iteration approximately. The authors show that $q$-means performs a robust version of $k$-means called $(\varepsilon,\tau)$-$k$-means, and then show through experiments that the approximation does not worsen the quality of the centroids. In the $(\varepsilon,\tau)$-$k$-means algorithm, two noise parameters $\varepsilon,\tau \geq 0$ are introduced in the distance estimation and centroid update steps. It alternates between the steps:
\begin{enumerate}[wide]
    \item A vector $v_i$ is assigned a cluster $\mathcal{C}_{\ell_i^t}^t$ where $\ell_i^t\in\{j\in[k]: \|v_i - c_j^t\|^2 \leq \min_{j'\in[k]}\|v_i - c_{j'}^t\|^2 + \tau\}$;
    \item The new centroids $\{c_j^{t+1}\}_{j\in[k]}$ are updated such that $\big\|c_j^{t+1} - \frac{1}{|\mathcal{C}_j^t|}\sum_{i\in\mathcal{C}_j^t}v_i\big\| \leq \varepsilon$.
\end{enumerate}
%
%Again, it is iterated until convergence.

The overall idea behind $q$-means from~\cite{kerenidis2019q} is to first create the quantum state $n^{-1/2}\sum_{i\in[n]} |i\rangle|\ell_i^t\rangle$ using distance estimation and quantum minimum finding~\cite{durr1996quantum}, followed by measuring the label register $|\ell_i^t\rangle$ to obtain $|\chi_j^t\rangle = |\mathcal{C}_j^t|^{-1/2}\sum_{i\in\mathcal{C}_j^t}|i\rangle$ for some random $j\in [k]$. %Here $|\chi_j^t\rangle$ is the quantum state of the characteristic vector $\chi_j^t\in\mathbb{R}^n$ defined as $(\chi_{j}^{t})_{i} = |\mathcal{C}_j^t|^{-1}$ if $i\in \mathcal{C}^t_{j}$ and $0$ if $i \not \in \mathcal{C}^t_{j}$. 
The $q$-means algorithm then proceeds to perform a matrix multiplication with matrix $V$ and the quantum state $|\chi_j^t\rangle$ by using quantum linear algebra techniques. This leads to $|c_j^{t+1}\rangle = \|c_j^{t+1}\|^{-1}\sum_{i\in[d]} (c_j^{t+1})_i |i\rangle $. %by the crucial fact that $c_j^{t+1} = V \chi_j^t$. 
Finally, quantum tomography~\cite{kerenidis2020quantuma} is employed on $|c_j^{t+1}\rangle$ in order to retrieve a classical description. Since quantum linear algebra techniques are used to obtain $|c_j^{t+1}\rangle$ from $|\chi_j^t\rangle$, the final runtime depends on quantities like the condition number $\kappa(V)$ of the matrix $V$ and a matrix dependent parameter $\mu(V)$ which is upper-bounded by the Frobenius norm $\|V\|_F$ of $V$.
\begin{fact}[{\cite[Theorem~3.1]{kerenidis2019q}}]
    For $\varepsilon>0$ and data matrix $V\in\mathbb{R}^{d\times n}$ with $1\leq \|v_i\|^2 \leq \nu$ and condition number $\kappa(V)$, the $q$-means algorithm outputs centroids consistent with the classical $(\varepsilon,\tau)$-$k$-means algorithm in $\widetilde{O}\big(kd\frac{\nu}{\varepsilon^2}\kappa(V)(\mu(V) + k\frac{\nu}{\tau}) + k^2\frac{\nu^{3/2}}{\varepsilon\tau}\kappa(V)\mu(V)\big)$ time per iteration.
\end{fact}

\subsection{Our results}

In this work, we provide \emph{exponentially} improved classical $(\varepsilon,\tau)$-$k$-means algorithms that match the logarithmic dependence on $n$ of the $q$-means algorithm from Kerenidis et al.~\cite{kerenidis2019q}. Our classical algorithms are based on $\ell_1$-sampling and assume the existence of data structures that allow efficient $\ell_1$-sampling from the columns of $V$ and from the vector of its column norms. We propose two classical algorithms that differentiate between the procedure for computing the labels $\ell_i^t$. In our first classical algorithm, the distances $\|v_i - c_j^t\|$ are exactly computed in time $O(d)$ each, yielding the exact labeling $\ell_i^t = \arg\min_{j\in[k]}\|v_i - c_j^t\|$. In our second classical algorithm, the distances $\|v_i - c_j^t\|$ are approximated up to error $\frac{\tau}{2}$ via an $\ell_2$-sampling procedure, which yields the approximate labeling $\ell_i^t \in \{j\in[k]:\|v_i - c_j^t\|^2 \leq \min_{j'\in[k]}\|v_i - c_{j'}^{t}\|^2 + \tau\}$.

\begin{table}[t]

\def\arraystretch{2}
%\resizebox{\linewidth}{!}{
\begin{minipage}[b]{\linewidth}
\centering
\begin{tabular}{|c|c|c|}
\hline
Algorithm  & \makecell{Errors} & Query complexity  \\ \hline\hline

\makecell{Classical \\ \Cref{alg:k-means}}   & $(\varepsilon,0)$ & $\widetilde{O}\Big(\Big(\frac{\|V\|^2}{n} + \frac{\|V\|^2_{2,1}}{n^2} \Big)\frac{k^2d}{\varepsilon^2} \Big)$ \\ \hline

\makecell{Classical \\ \Cref{alg:k-means2}} & $(\varepsilon,\tau)$ & $\widetilde{O}\Big(\Big(\frac{\|V\|^2}{n} + \frac{\|V\|_{1,1}^2}{n^2} \Big)\frac{\|V\|_F^2\|V\|_{2,\infty}^2}{n}\frac{k^3}{\varepsilon^2\tau^2} \Big)$ \\ \hline

\makecell{Quantum \\ \Cref{alg:q-means}} & $(\varepsilon,0)$ & $\widetilde{O}\Big(\Big(\sqrt{k}\frac{\|V\|}{\sqrt{n}} + \sqrt{d}\frac{\|V\|_{2,1}}{n}\Big)\frac{k^{\frac{3}{2}}d}{\varepsilon}\Big)$ \\ \hline

\makecell{Quantum \\ \Cref{alg:q-means2}} & $(\varepsilon,\tau)$ & $\widetilde{O}\Big(\Big(\sqrt{k}\frac{\|V\|}{\sqrt{n}} + \sqrt{d}\frac{\|V\|_{1,1}}{n} \Big)\frac{\|V\|_F\|V\|_{2,\infty}}{\sqrt{n}}\frac{k^{\frac{3}{2}}}{\varepsilon\tau} \Big)$ \\ \hline
\end{tabular}
\end{minipage}
\vspace{0.4cm}
\newline
\centering
\begin{tabular}{|c|c|c|}
\hline
Algorithm  & \makecell{Errors} & Time complexity  \\ \hline\hline

\makecell{Classical \\ \Cref{alg:k-means}}   & $(\varepsilon,0)$ & $\widetilde{O}\Big(\Big(\frac{\|V\|^2}{n} + \frac{\|V\|_{2,1}^2}{n^2}\Big)\frac{k^2d}{\varepsilon^2}(\log{n} + k)\Big)$ \\ \hline

\makecell{Classical \\ \Cref{alg:k-means2}} & $(\varepsilon,\tau)$ & $\widetilde{O}\Big(\Big(\frac{\|V\|^2}{n} + \frac{\|V\|_{1,1}^2}{n^2} \Big)\frac{\|V\|_F^2\|V\|_{2,\infty}^2}{n}\frac{k^3}{\varepsilon^2\tau^2}\log{n} \Big)$ \\ \hline

\makecell{Quantum \\ \Cref{alg:q-means}} & $(\varepsilon,0)$ & $\widetilde{O}\Big(\Big(\sqrt{k}\frac{\|V\|}{\sqrt{n}} + \sqrt{d}\frac{\|V\|_{2,1}}{n}\Big)\frac{k^{\frac{3}{2}}d}{\varepsilon}(\log{n} + \sqrt{k})\Big)$ \\ \hline

\makecell{Quantum \\ \Cref{alg:q-means2}} & $(\varepsilon,\tau)$ & $\widetilde{O}\Big(\Big(\sqrt{k}\frac{\|V\|}{\sqrt{n}} + \sqrt{d}\frac{\|V\|_{1,1}}{n} \Big)\Big(\frac{\|V\|_F\|V\|_{2,\infty}}{\sqrt{n}}\frac{k^{\frac{3}{2}}}{\varepsilon\tau}\log{n} + \frac{k^2d}{\varepsilon}\Big) \Big)$ \\ \hline
\end{tabular}
\caption{Query and time complexities per iteration of our classical and quantum algorithms assuming $|\mathcal{C}_j^t| = \Omega(\frac{n}{k})$ for all $j\in[k]$. The errors $(\varepsilon,\tau)$ refer to the error $\tau$ in assigning a vector $v_i$ to a cluster $\mathcal{C}^t_{\ell_i^t}$ with $\ell_i^t \in \{j\in[k]:\|v_i - c_j^t\|^2 \leq \min_{j'\in[k]}\|v_i - c_{j'}^{t}\|^2 + \tau\}$, and the error $\varepsilon$ in computing the new centroids as $\big\|c_j^{t+1} - |\mathcal{C}_j^t|^{-1}\sum_{i\in\mathcal{C}_j^t}v_i\big\| \leq \varepsilon$. The matrix norms are $\|V\| = \max_{x\in\mathbb{R}^d:\|x\|=1}\|Vx\|$, $\|V\|_F = (\sum_{i\in[n],l\in[d]} V_{li}^2)^{\frac{1}{2}}$, $\|V\|_{1,1} = \sum_{i\in[n],l\in[d]} |V_{li}|$, $\|V\|_{2,1} = \sum_{i\in[n]}\|v_i\|$, $\|V\|_{2,\infty} = \max_{i\in[n]}\|v_i\|$. The quantum runtimes can be slightly improved if one has access to a special gate called QRAG~\cite{ambainis2007quantum,Allcock2024constantdepth}, see \Cref{footnote3,footnote4}. Here $\widetilde{O}(\cdot)$ hides $\polylog$ factors in $k$, $d$, $\frac{1}{\varepsilon}$, $\frac{1}{\tau}$, $\frac{\|V\|_F}{\sqrt{n}}$.}
\label{table:results}
\end{table}

Furthermore, we also propose improved \emph{quantum} algorithms that avoid the need for quantum linear algebra subroutines and still keep the logarithmic dependence on the size $n$ of the dataset, while \emph{improving} the complexity of our classical $(\varepsilon,\tau)$-$k$-means algorithms in several parameters. Similar to~\cite{kerenidis2019q}, our quantum algorithms construct states of the form $\sum_{i\in[n]}\frac{1}{\sqrt{n}} |i,\ell_i^t\rangle$ and $\sum_{i\in[n]}\sqrt{\frac{\|v_i\|}{\|V\|_{2,1}}}|i,\ell_i^t\rangle$, where $\|V\|_{2,1} = \sum_{i\in[n]}\|v_i\|$, but instead of performing quantum linear algebra subroutines and quantum tomography, we input such states (or more precisely the unitaries behind them) into a quantum amplitude estimation subroutine, in this case, the multivariate quantum Monte Carlo subroutine from Cornelissen, Hamoudi, and Jerbi~\cite{cornelissen2022near}. The outcome is several vectors that approximate $c_1^{t+1},\dots,c_k^{t+1}$ up to an additive error $\varepsilon$ in $\ell_\infty$-norm. In order to construct a superposition of labels $|\ell_i^t\rangle$, we require $V$ to be stored in a KP-tree~\cite{prakash2014quantum,kerenidis2016quantum} and to be accessible via a QRAM~\cite{giovannetti2008architectures,giovannetti2008quantum}. Once again, we propose two quantum algorithms that differentiate between the procedure for computing the labels $\ell_i^t$. In our first quantum algorithm, the distance $\|v_i - c_j^t\|$ are exactly computed in time $O(d)$, while in our second quantum algorithm, such distances are approximated using quantum amplitude estimation~\cite{brassard2002quantum}. Finally, we provide classical and quantum query lower bounds that show that our algorithms are optimal in most parameters. Our exact query and time complexities are described in \Cref{table:results}. In the next result, we provide a simplified version of our query complexities. Let $\|V\|_{2,\infty} = \max_{i\in[n]}\|v_i\|$ and $\|V\|_{1,1} = \sum_{i\in[n],l\in[d]}|V_{li}|$.

\begin{result}\label{thr:result}
    Let $\varepsilon,\tau > 0$ and $\delta\in(0,1)$. Assume sampling and QRAM access to $V\in\mathbb{R}^{d\times n}$ and $(\|v_i\|)_{i\in[n]}$. Assume all clusters satisfy $|\mathcal{C}_j^t| = \Omega(\frac{n}{k})$. There are classical and quantum algorithms that output centroids consistent with $(\varepsilon,\tau)$-$k$-means with probability $1-\delta$ and with per-iteration query complexity (up to $\polylog$ factors in $k$, $d$, $\frac{1}{\varepsilon}$, $\frac{1}{\tau}$, $\frac{1}{\delta}$, $\frac{\|V\|_F}{\sqrt{n}}$)
    \begin{align*}
        &\textbf{Classical:}~\widetilde{O}\bigg({\min}\bigg\{\frac{\|V\|_F^2}{n}\frac{k^2d}{\varepsilon^2}, ~\bigg(\frac{\|V\|^2_F}{n} + \frac{\|V\|_{1,1}^2}{n^2} \bigg)\frac{\|V\|_F^2\|V\|_{2,\infty}^2}{n}\frac{k^3}{\varepsilon^2\tau^2}\bigg\}\bigg);\\
        &\textbf{Quantum:}~ \widetilde{O}\bigg({\min}\bigg\{\frac{\|V\|_F}{\sqrt{n}}\frac{k^{\frac{3}{2}}d}{\varepsilon}(\sqrt{k} + \sqrt{d}), ~\Big(\sqrt{k}\frac{\|V\|_F}{\sqrt{n}} + \sqrt{d}\frac{\|V\|_{1,1}}{n} \Big)\frac{\|V\|_F\|V\|_{2,\infty}}{\sqrt{n}}\frac{k^{\frac{3}{2}}}{\varepsilon\tau}\bigg\}\bigg).
    \end{align*}

    Moreover, with entry-wise query access to $V \in \mathbb{R}^{d\times n}$ and $(\|v_i\|)_{i\in[n]}$ and classical description of partition $\{\mathcal{C}_j\}_{j\in[k]}$ of $[n]$, any classical or quantum algorithm that outputs $c_1,\dots,c_k\in\mathbb{R}^d$ with $\big\|c_j - \frac{1}{|\mathcal{C}_j|}\sum_{i\in\mathcal{C}_j}v_i\big\| \leq \varepsilon$ has query complexity
    \begin{align*}
        \textbf{Classical:}~\Omega\bigg({\min}\bigg\{\frac{\|V\|_F^2}{n}\frac{kd}{\varepsilon^2}, nd\bigg\}\bigg); \qquad\textbf{Quantum:}~\Omega\bigg({\min}\bigg\{\frac{\|V\|_F}{\sqrt{n}}\frac{kd}{\varepsilon}, nd\bigg\} \bigg).
    \end{align*}
\end{result}

Our classical $(\varepsilon,\tau)$-$k$-means algorithm can be seen as a ``dequantised'' version of our improved $q$-means quantum algorithm  ~\cite{tang2019quantum,tang2021quantum,gilyen2018quantum,gilyen2022improved} and shows that, given appropriate access to the data matrix, a logarithmic dependence on the data size $n$ can be achieved, exponentially better than prior works. Quantum computers can yield quadratic improvements in several parameters like $\varepsilon$ and $\frac{\|V\|}{\sqrt{n}}$, as evident from our quantum upper bound. Finally, our query lower bounds show that our algorithms are optimal with respect to $\varepsilon$ and $\frac{\|V\|_F}{\sqrt{n}}$, and only the dependence on $d$ and $k$ are slightly sub-optimal.

\subsection{Related independent work and future directions}

We briefly mention related works that have appeared online around the same time or later than ours. First, the independent work of Jaiswal~\cite{Jaiswal2023kmeans} (see also~\cite{shah2025quantum}) quantized the highly parallel, sampling-based approximation scheme of~\cite{bhattacharya_et_al:LIPIcs:2020:13254} and thus obtained a quantum algorithm for the $k$-means problem with \emph{provable} guarantees, as opposed to our results, which are heuristic. More specifically, their quantum algorithm outputs a centroid set that minimizes the $k$-means cost function up to relative error $\varepsilon$. Due to such a guarantee, though, their final runtime depends exponentially in $k$ and $\frac{1}{\varepsilon}$ (but maintains the polylogarithmic dependence on $n$). Another related work is from Kumar et al.~\cite{Kumar2025desqquantum}, who introduced a quantum algorithm called Des-q based on the original $q$-means algorithm from Kerenidis et al.~\cite{kerenidis2019q}. Their quantum algorithm constructs and retrains decision trees for regression and binary classification tasks, while achieving a logarithmic complexity in the combined total number of old and new data vectors, even accounting for the time needed to load the new samples into QRAM. 

Very recently, Chen et al.~\cite{chen2025provablyfasterrandomizedquantum} proposed alternative algorithms to the ones presented here by employing uniform sampling plus shifting the vectors $v_i$ by the current centroids $c_1^t,\dots,c_k^t$. Their algorithm uses $\widetilde{O}\big(k^{\frac{5}{2}}\sqrt{d}\big(\frac{\sqrt{\phi}}{\varepsilon} + \sqrt{d}\big) \big)$ QRAM calls, where $\phi := \frac{1}{n}\sum_{j\in[k]}\sum_{i\in\mathcal{C}_j^{t}}\|v_i - c_j^{t+1}\|^2$ and $c_j^{t+1} = |\mathcal{C}_j^{t}|^{-1}\sum_{i\in\mathcal{C}_j^{t}}v_i$ for clusters $\{\mathcal{C}_j^{t}\}_{j\in[k]}$ defined by $c_1^t,\dots,c_k^t$. Since one can write $\phi = \frac{1}{n}\sum_{j\in[k]}\big(\sum_{i\in\mathcal{C}_j^{t}}\|v_i\|^2 - |\mathcal{C}_j^{t}|^{-1}\|\sum_{i\in\mathcal{C}_j^{t}}v_i\|^2\big)$, then $\sqrt{\phi} \leq \frac{\|V\|_F}{\sqrt{n}}$. Our quantum algorithm, on the other hand, makes $\widetilde{O}\big(\big(\sqrt{k}\frac{\|V\|}{\sqrt{n}} + \sqrt{d}\frac{\|V\|_{2,1}}{n}\big)\frac{k^{3/2}}{\varepsilon} \big)$ QRAM calls assuming $|i\rangle|\bar{0}\rangle \mapsto |i\rangle|v_i\rangle$ counts as $1$ QRAM call as in~\cite{chen2025provablyfasterrandomizedquantum} (for the remaining of the paper we shall assume that this map counts as $d$ QRAM calls instead). If $\sqrt{\phi} \ll \frac{\|V\|}{\sqrt{n}} + \frac{\|V\|_{2,1}}{n}$ (note that $\|V\|\leq \|V\|_F$ and $\frac{\|V\|_{2,1}}{n} \leq \frac{\|V\|_F}{\sqrt{n}}$), e.g., when the distance between vectors within the same cluster is much smaller than the distance between clusters, the complexity from Chen et al.~\cite{chen2025provablyfasterrandomizedquantum} can be better than ours, otherwise, if $\sqrt{\phi} = \omega\big(\frac{1}{\sqrt{kd}}\frac{\|V\|}{\sqrt{n}} + \frac{1}{k}\frac{\|V\|_{2,1}}{n}\big)$, our complexity is better. Both results are thus incomparable. Finally, our classical and quantum lower bounds readily imply $\Omega\big({\min}\big\{\frac{\phi kd}{\varepsilon^2},nd\big\}\big)$ and $\Omega\big({\min}\big\{\frac{\sqrt{\phi} kd}{\varepsilon},nd\big\}\big)$, respectively (assuming that reading $v_i$ takes $d$ queries).

We mention a few future directions. One is to consider a quantum version of $k$-means++ using similar techniques employed here or assume data vectors with special properties, e.g., well-clusterable datasets~\cite{kerenidis2019q}, in order to obtain tighter runtimes. In this direction, see Chen et al.~\cite{chen2025provablyfasterrandomizedquantum} who exploited certain symmetries of $k$-means. Another interesting direction is to bridge both our and Jaiswal~\cite{Jaiswal2023kmeans} works and obtain improved time complexities in $k$ and $\frac{1}{\varepsilon}$ together with provable guarantees.

\section{Preliminaries}

Let $[n] \!=\! \{1,\dots,n\}$ for $n\in\mathbb{N} \!=\! \{1,2,\dots\}$. %In the rest of the paper, our dataset is a matrix $V = [v_1,\dots,v_n] \in \R^{n \times d}$ comprised of row vectors $v_i\in\mathbb{R}^d$, $i\in[n]$, such that $\min_{i\in[n]}\|v_i\| \geq 1$, where $n$ is the number of samples and $d$ is the number of features. 
For $v\in\mathbb{R}^d$, let $\|v\| \!=\! (\sum_{i\in[d]} v_i^2)^{\frac{1}{2}}$, $\|v\|_1 \!=\! \sum_{i\in[d]}|v_i|$, and $\|v\|_\infty \!=\! \max_{i\in[d]}|v_i|$. Let $\mathcal{D}^{(1)}_v$ and $\mathcal{D}^{(2)}_v$ be the distributions over $[d]$ with probability density functions $\mathcal{D}^{(1)}_v(i) = \frac{|v_i|}{\|v\|_1}$ and $\mathcal{D}^{(2)}_v(i) = \frac{v_i^2}{\|v\|^2}$. For $V\in\mathbb{R}^{d\times n}$, let the matrix~norms:
\begin{itemize}
    \begin{minipage}{0.65\linewidth}  
    \item $\|V\| = \max_{x\in\mathbb{R}^n:\|x\|=1}\|Vx\|$ (spectral norm);
    \item $\|V\|_F = (\sum_{i\in[n]}\sum_{l\in[d]} V_{li}^2)^{\frac{1}{2}}$ (Frobenius norm);
    \item $\|V\|_{1,1} = \sum_{i\in[n]}\sum_{l\in[d]} |V_{li}| = \sum_{i\in[n]}\|v_i\|_1$;
    \end{minipage}
    \begin{minipage}{0.34\linewidth}
    \item $\|V\|_{2,1} = \sum_{i\in[n]}\|v_i\|$;
    \item $\|V\|_{2,\infty} = \max_{i\in[n]}\|v_i\|$.
    \item[\vspace{\fill}]
    \item[\vspace{\fill}]
    \end{minipage}
\end{itemize}
It is known that $\|V\|_{2,\infty} \leq \|V\| \leq \|V\|_F \leq \|V\|_{2,1} \leq \|V\|_{1,1}$ and $\|V\|_F \leq \sqrt{\min\{n,d\}}\|V\| \leq \sqrt{n\min\{n,d\}}\|V\|_{2,\infty}$ and $\|V\|_{1,1} \leq \sqrt{d}\|V\|_{2,1} \leq \sqrt{nd}\|V\|_F \leq n\sqrt{d}\|V\|_{2,\infty}$.
Let $\mathcal{D}^{(1)}_V$ be the distribution over $[n]\times[d]$ with probability density function $\mathcal{D}^{(1)}_V(i,l) = \frac{|V_{li}|}{\|V\|_{1,1}}$. We use $|\bar{0}\rangle$ to denote the state $|0\rangle\otimes\cdots\otimes|0\rangle$ where the number of qubits is clear from context.

% Given two probability distributions $P,Q$ over a discrete sample space $\mathcal{X}$, their total variation distance is $\|P-Q\|_{\rm tvd} := \frac{1}{2}\sum_{x\in\mathcal{X}}|P(x) - Q(x)|$. 
In the next sections, we will use the quantum minimum finding subroutine from D\"urr and H\o{}yer~\cite{durr1996quantum}, its generalisation with variable times due to Ambainis~\cite{ambainis2010quantum,ambainis2012variable}, and the multivariate quantum Monte Carlo subroutine from Cornelissen, Hamoudi, and Jerbi~\cite{cornelissen2022near}.
\begin{fact}[Quantum min-finding~\cite{durr1996quantum}]
    \label{fact:minimum_finding}
    Given $\delta\in(0,1)$ and oracle $U_x:|i\rangle|\bar{0}\rangle \mapsto |i\rangle|x_i\rangle$ for $x\in\mathbb{R}^N$, there is a quantum algorithm that outputs $|\Psi_x\rangle$ using $O(\sqrt{N}\log\frac{1}{\delta})$ queries to $U_x$ such that, upon measuring $|\Psi_x\rangle$ in the computational basis, the outcome is $\arg\min_{i\in[N]}x_i$ with probability $1-\delta$.
\end{fact}
%
%\textcolor{red}{Joao: I am not entirely sure that Ambainis' paper is doing what I wrote below}
\begin{fact}[Variable-time quantum min-finding~\cite{ambainis2010quantum}]\label{fact:variable_time_minimum_finding}
    Let $\delta\in(0,1)$, $x\in\mathbb{R}^N$, and $\{U_i\}_{i\in[N]}$ a collection of oracles such that $U_{i}:|\bar{0}\rangle \mapsto |x_i\rangle$ in time $O(t_i)$. There is a quantum algorithm that runs in time $O((\sum_{i\in[N]} t_i^2)^{\frac{1}{2}}\log\frac{1}{\delta})$ and outputs $|\Psi_x\rangle$ such that, upon measuring $|\Psi_x\rangle$ in the computational basis, the outcome is $\arg\min_{i\in[N]}x_i$ with probability $1-\delta$.
\end{fact}

\begin{fact}[{\cite[Theorem~3.3]{cornelissen2022near}}]\label{fact:multivariate_Monte_Carlo}
    Consider a bounded $N$-dimensional random variable $X:\Omega\to\mathbb{R}^N$ over a probability space $(\Omega,2^\Omega,\mathbb{P})$ with mean $\mu = \sum_{\omega\in\Omega}\mathbb{P}(\omega)X(\omega)$ and such that $\|X\| \leq 1$. Assume access to unitaries $U_{\mathbb{P}}:|\bar{0}\rangle \mapsto \sum_{\omega\in\Omega}\sqrt{\mathbb{P}(\omega)}|\omega\rangle$ and $\mathcal{B}_{X} : |\omega\rangle|\bar{0}\rangle \mapsto |\omega\rangle|X(\omega)\rangle$. Given $\delta\in(0,1)$, $m\in\mathbb{N}$, and an upper bound $L_2 \geq \mathbb{E}[\|X\|]$, there is a quantum algorithm that outputs $\widetilde{\mu}\in\mathbb{R}^N$ such that $\|\mu - \widetilde{\mu}\|_\infty \leq \frac{\sqrt{L_2}\log(N/\delta)}{m}$ with success probability at least $1-\delta$, using $O(m\poly\log{m})$ queries to the oracles $U_{\mathbb{P}}$ and $\mathcal{B}_{X}$, and in time $\widetilde{O}(mN)$.\footnote{The time complexity of the mean estimation subroutine is not analysed in \cite{cornelissen2022near}, so we give a sketch of its analysis here. The last step of the algorithm needs to perform $N$ parallel inverse QFTs on $m$ qubits, which requires $\widetilde{\Theta}(mN)$ gates since we need to touch every qubit at least once. It remains to show that we can implement the rest of the algorithm in time $\widetilde{O}(mN)$. In the preprocessing step, we compute the $\ell_2$-norm of the random variable in time $O(N)$, a total of $\widetilde{O}(m)$ times. The main routine, subsequently, runs with $m$ repetitions, within each of which we perform arithmetic operations such as computing the inner product of two $N$-dimensional vectors, in $O(N)$ time, and do quantum singular value transformations. This last step is used to turn a so-called probability oracle into a phase oracle, and takes $\widetilde{O}(1)$ time to implement. The total time complexity of this step thus also becomes $\widetilde{O}(mN)$.\label{footnote1}}
\end{fact}

Finally, the following concentration inequalities and approximation lemma will be useful.

\begin{fact}[Chernoff's bound]\label{fact:chernoff}
    Let $X:= \sum_{i\in[N]} X_i$ where $X_1,\dots,X_N$ are independently distributed in $[0,1]$. Then $\operatorname{Pr}[|X - \mathbb{E}[X]| \geq \epsilon\mathbb{E}[X]] \leq 2e^{-\epsilon^2 \mathbb{E}[X]/3}$ for all $\epsilon>0$.
\end{fact}

\begin{fact}[Median-of-means {\cite[Proposition 12]{lerasle2019lecture}}]\label{fact:median-of-means}
    Let $X^{(j)} := \frac{1}{N}\sum_{i\in[N]} X_i^{(j)}$ for $j\in[K]$, where $\{X_i^{(j)}\}_{i\in[N],j\in[K]}$ are i.i.d.\ copies of a random variable $X$.  Then, for all $\epsilon > 0$ and $\sigma^2 \geq \operatorname{Var}(X)$,
    \begin{align*}
        {\operatorname{Pr}}\big[|{\operatorname{median}}(X^{(1)},\dots,X^{(K)}) - \mathbb{E}[X]| \geq \epsilon\big] \leq \exp\bigg(-2K\bigg(\frac{1}{2} -\frac{\sigma^2}{N\epsilon^2} \bigg)^2 \bigg).
    \end{align*}
\end{fact}

\begin{fact}[Freedman's inequality {\cite[Theorem~1.6]{freedman1975tail} \& \cite[Theorem~1.1]{tropp2011freedman}}]
    \label{thm:freedman_inequality}
    Let $\{Y_i:i\in\mathbb{N}\cup\{0\}\}$ be a real-valued martingale with difference sequence $\{X_i:i\in\mathbb{N}\}$. Assume that $X_i \leq B$ almost surely for all $i\in\mathbb{N}$. Let $W_i := \sum_{j\in[i]} \mathbb{E}[X_j^2|X_1,\dots,X_{j-1}]$ for $i\in\mathbb{N}$. Then, for all $\epsilon\geq 0$ and $\sigma>0$,
    \begin{align*}
        \operatorname{Pr}[\exists i \geq 0: Y_i\geq \epsilon ~\text{and}~ W_i \leq \sigma^2] \leq \exp\left(-\frac{\epsilon^2/2}{\sigma^2 + B \epsilon/3}\right).
    \end{align*}
    
\end{fact}

\begin{claim}\label{lem:error_propagation}
    Let $\widetilde{a},a\in\mathbb{R}_+$ be such that $|a - \widetilde{a}| \leq \epsilon$, where $\epsilon\in[0,\frac{a}{2}]$. Then $\left|\frac{1}{\widetilde{a}} - \frac{1}{a} \right| \leq \frac{2\epsilon}{a^2}$.
\end{claim}
% \begin{proof}
%     Since $a - \widetilde{a} \leq \epsilon \implies \frac{1}{\widetilde{a}} \leq \frac{1}{a - \epsilon} \leq \frac{2}{a}$, then $\left|\frac{1}{\widetilde{a}} - \frac{1}{a} \right| =  \left|\frac{\widetilde{a} - a}{a\widetilde{a}} \right| \leq \frac{2\epsilon}{a^2}$.
% \end{proof}

\subsection{Computational models and data structures} 

In our classical and quantum algorithms, we shall assume the existence of a low-overhead classical or quantum data structure storing the dataset matrix $V\in\mathbb{R}^{d\times n}$ and the centroid matrix $C\in\mathbb{R}^{d\times k}$.
\begin{definition}[Classical sampling access]
    \label{lem:matrix_structure}
    We say we have \emph{(classical) sampling access} to a matrix $V=[v_1,\dots,v_n]\in\mathbb{R}^{d\times n}$ if $V$ is stored in a data structure that supports the following operations:
    \begin{enumerate}[itemsep=0pt,topsep=1pt]
        \item Reading an entry of $V$ in $O(\log(nd))$ time;
        \item Finding $\|v_i\|$ and $\|V\|_{2,1}$ in $O(\log{n})$ and $O(1)$ time, respectively;
        \item Sampling from $\mathcal{D}^{(1)}_{(\|v_i\|)_{i=1}^n}(i) = \frac{\|v_i\|}{\|V\|_{2,1}}$ in $O(\log{n})$ time;
        \item Sampling from $\mathcal{D}^{(2)}_{v_i}(l) = \frac{V_{li}^2}{\|v_i\|^2}$ and $\mathcal{D}^{(1)}_V(i,l) = \frac{|V_{li}|}{\|V\|_{1,1}}$ in $O(\log(nd))$ time.
    \end{enumerate}
\end{definition}
\begin{definition}[Quantum query access]
    \label{lem:kp_tree}
    We say we have \emph{quantum query access} to a matrix $V=[v_1,\dots,v_n]\in\mathbb{R}^{d\times n}$ if $V$ is stored in a data structure that supports the following operations:
    \begin{enumerate}[itemsep=0pt,topsep=1pt]
        \item Reading an entry of $V$ in $O(\log(nd))$ time; \label{item:item0}
        \item Finding $\|v_i\|$ and $\|V\|_{2,1}$ in $O(\log{n})$ and $O(1)$ time, respectively; \label{item:item3}
        \item Mapping $|\bar{0}\rangle \mapsto \sum_{i\in[n]} \sqrt{\frac{\|v_i\|}{\|V\|_{2,1}}} |i\rangle$ in $O(\log{n})$ time; \label{item:item5}
        \item Mapping $|\bar{0}\rangle \mapsto \sum_{(i,l)\in[n]\times[d]} \sqrt{\frac{|V_{li}|}{\|V\|_{1,1}}} |i,l\rangle$ in $O(\log(nd))$ time; \label{item:item5a}
        \item Mapping $|i\rangle|\bar{0}\rangle \mapsto \sum_{l\in[d]} \frac{V_{li}}{\|v_i\|}|i,l\rangle$ in $O(\log(nd))$ time;\label{item:item5b}
        \item Mapping $|i,l\rangle|\bar{0}\rangle \mapsto |i,l\rangle|V_{li}\rangle$ in $O(\log(nd))$ time; \label{item:item1}
        \item Mapping $|i\rangle|\bar{0}\rangle \mapsto |i\rangle|v_i\rangle = |i\rangle|V_{1i},\dots,V_{di}\rangle$ in $O(d\log{n})$ time. \label{item:item2}
    \end{enumerate}
\end{definition}

In our classical computation model, all arithmetic operations require $O(1)$ time. We refer to any operation in \Cref{lem:matrix_structure} as a \emph{classical query}. We define a measure of (classical) \emph{time complexity} as the sum of the times of all arithmetic and classical queries comprising some given computation. In other words, if a computation is composed of $m$ arithmetic operations and $q$ classical queries, its time complexity is at most $O(m + q\log(nd))$.

In our quantum computation model, all quantum gates require $O(1)$ time, and more generally, all arithmetic operations require $O(1)$ time. We refer to any operation in \Cref{lem:kp_tree} as a \emph{quantum query}, except \Cref{item:item2} which comprises $d$ quantum queries. We define a measure of (quantum) \emph{time complexity} as the sum of the times of all arithmetic and quantum queries comprising some given computation. In other words, if a computation is composed of $m$ arithmetic operations and $q$ quantum queries, its time complexity is at most $O(m + q\log(nd))$.

In \Cref{lem:kp_tree}, the unitary from \Cref{item:item1,item:item2} is known as a quantum random access memory (QRAM)~\cite{giovannetti2008architectures,giovannetti2008quantum,hann2021practicality,phalak2023quantum,jaques2023qram,Allcock2024constantdepth} and can be seen as the quantum equivalent of a classical RAM. %Here by $|v_{i}\rangle$ we mean a \emph{binary encoding} of each entry of $v_i$, and not an \emph{amplitude encoding} as normally assumed in other works. 
%Several different QRAM architectures and circuits have been proposed~\cite{}. 
In this work we simply assume that a QRAM requires a running time proportional to its circuit depth, i.e., logarithmitically in the memory size it accesses. Finally, the classical data structure that allows the operations from \Cref{item:item0,item:item3,item:item5,item:item5a,item:item5b}, called KP-tree, was proposed by Kerenidis and Prakash~\cite{prakash2014quantum,kerenidis2016quantum}.

%in order to implement is a quantum circuit that performs $O(\log{d})$ queries to an oracle called QRAM (quantum random access memory). This oracle allows query of the form $\ket{i}\ket{0} \mapsto \ket{i}\ket{m_i}$ where $m_i \in \{0,1\}^m$, and $i \in [L]$ with a circuit depth of $O(\log L)$. This is similar to the memory of a classical computer. Importantly, working in the QRAM model with the oracles defined in \Cref{lem:kp_tree} assumes the existence of the mapping $ \ket{i,j}\mapsto \ket{i,j}\ket{V_{ij}},$ where $V_{ij}$ is the entry of the matrix $V$ at row $i$ and column $j$. The importance of this model consist in the existence of circuits of depth that is polylogarithmic in the size of the matrix. In a computational model where the run time of the quantum circuit is proportional to the depth, the QRAM model allows to obtain quantum algorithms whose runtime is sub-linear in the size of the matrix. For a precise (but verbose) definition of quantum computer with quantum access to classical memory, the interested reader is referred to \cite[Definition 8]{allcock2023constant}. 

\section{Classical algorithms}

We now present our classical $(\varepsilon,\tau)$-$k$-means algorithms, whose main idea is to employ the classical sampling access from~\Cref{lem:matrix_structure} to separately estimate the quantities $\sum_{i\in\mathcal{C}_j^t} v_i$ and $|\mathcal{C}_j^t|$ for $j\in[k]$, from which the new centroids $c_j^{t+1} \!\approx\! |\mathcal{C}_j^t|^{-1}\sum_{i\in\mathcal{C}_j^t} v_i$ are approximated. For our first algorithm, we sample columns of $V$ with probability $\frac{\|v_i\|}{\|V\|_{2,1}}$ and select those closer to $c_j^t$ to approximate $\sum_{i\in\mathcal{C}_j^t} v_i$. To approximate $|\mathcal{C}_j^t|$ we sample columns of $V$ uniformly at random instead.

\begin{theorem}[Classical $(\varepsilon,0)$-$k$-means algorithm]\label{thr:classical_algorithm}
    Let $\varepsilon>0$, $\delta\in(0,1)$, and assume sampling access to $V=[v_1,\dots,v_n]\in\mathbb{R}^{d\times n}$. If all clusters satisfy $|\mathcal{C}_j^t| = \Omega(\frac{n}{k})$, then {\rm \Cref{alg:k-means}} outputs centroids consistent with the $(\varepsilon,\tau=0)$-$k$-means algorithm with probability $1-\delta$. The complexities per iteration of {\rm \Cref{alg:k-means}} are
    \begin{align*}
        \textbf{Classical queries:}&~O\bigg(\bigg(\frac{\|V\|^2}{n} + \frac{\|V\|_{2,1}^2}{n^2}\bigg)\frac{k^2d}{\varepsilon^2}\log\frac{k}{\delta}\bigg), \\ % = O\left(\frac{\|V\|_F^2}{n}\frac{k^2}{\varepsilon^2}\log\frac{k}{\delta}\right), \\
        \textbf{Time:}&~O\bigg(\bigg(\frac{\|V\|^2}{n} + \frac{\|V\|_{2,1}^2}{n^2}\bigg)\frac{k^2d}{\varepsilon^2}(k + \log(nd))\log\frac{k}{\delta}\bigg).% = O\left(\frac{\|V\|_F^2}{n}\frac{k^2}{\varepsilon^2}(kd + \log{n})\log\frac{k}{\delta}\right).
    \end{align*}
\end{theorem}
\begin{proof}
    Let $\chi_j^t\in\mathbb{R}^n$ be the characteristic vector for cluster $j \in [k]$ at iteration $t$ scaled to $\ell_1$-norm, i.e., $(\chi_{j}^{t})_{i} = \frac{1}{|\mathcal{C}_j^t|}$ if $i\in \mathcal{C}^t_{j}$ and $0$ if $i \not \in \mathcal{C}^t_{j}$. For $i\in[n]$, let $\ell_i^t := \arg\min_{j\in[k]}\|v_i - {c}_j^t\|$. Sample $p$ indices $P\subseteq[n]$ uniformly from $[n]$. Let $\lambda_j = \frac{p}{|P_j|}\frac{|\mathcal{C}_j^t|}{n}$, where $P_j := \{i\in P | \ell_i^t = j\}$. On the other hand, sample $q$ indices $Q\subseteq[n]$ from the distribution $\mathcal{D}^{(1)}_{(\|v_i\|)_{i=1}^n}(i) = \frac{\|v_i\|}{\|V\|_{2,1}}$. For each $i\in Q$, let $X_i\in\mathbb{R}^{d\times n}$ be the matrix formed by setting the $i$-th column of $X$ to $\|V\|_{2,1}\frac{v_i}{\|v_i\|}$ and the rest to zero. Define $\widetilde{V} := \frac{1}{q}\sum_{i\in Q} X_i$. Then $\mathbb{E}[\widetilde{V}] = V$. 
\begin{algorithm}[t]
    \caption{Classical $(\varepsilon,\tau=0)$-$k$-means algorithm}
    \begin{algorithmic}[1]
    \Require Sampling access to data matrix $V=[v_1,\dots,v_n]\in\mathbb{R}^{d\times n}$, parameters $\delta,\varepsilon$.
%    \Ensure Vectors $c_1,\dots,c_k\in\mathbb{R}^d$ corresponding to centroids.
    \State Select $k$ initial centroids ${c}_1^0,\dots,{c}_k^0$ %and build sampling access to $C^0=[c_1^0,\dots,c_k^0]\in\mathbb{R}^{k\times d}$
    \For{$t=0$ until convergence}
        \State Sample $p = O\big(\frac{\|V\|^2}{n}\frac{k^2}{\varepsilon^2}\log\frac{k}{\delta}\big)$ indices $P\subseteq[n]$ uniformly from $[n]$
        \State Sample $q = O\big(\frac{\|V\|_{2,1}^2}{n^2}\frac{k^2}{\varepsilon^2}\log\frac{k}{\delta}\big)$ indices $Q\subseteq[n]$ from $\mathcal{D}^{(1)}_{(\|v_i\|)_{i=1}^n}$
        \State For $i\in P\cup Q$, label $\ell_i^t = \arg\min_{j\in[k]}\|v_{i} - {c}_j^t\|$. 
        \State For $j\in[k]$, let $P_j := \{i\in P|\ell^t_i = j\}$ and $Q_j := \{i\in Q|\ell_i^t = j\}$
        \State For $j\in[k]$, let the new centroids ${c}_j^{t+1} = \frac{p}{n|P_j|}\sum_{i\in Q_j}\frac{\|V\|_{2,1}}{q}\frac{v_{i}}{\|v_{i}\|}$
    \EndFor
\end{algorithmic}\label{alg:k-means}
\end{algorithm}
    
    We start with the error analysis. We note that the outputs of the standard $k$-means and \Cref{alg:k-means} can be stated, respectively, as $c_j^{\ast\:t+1} = V\chi_j^t$ and $c_j^{t+1} = \lambda_j \widetilde{V}\chi_j^{t} = \frac{p}{n|P_j|}\sum_{i\in Q_j}\frac{\|V\|_{2,1}}{q}\frac{v_i}{\|v_i\|}$, where $Q_j := \{i\in Q|\ell_i^t = j\}$. In order to bound $\|c_j^{\ast\:t+1} - c_j^{t+1}\|$, first note that, by the triangle inequality, 
    \begin{align*}
        \|c_j^{\ast\:t+1} - c_j^{t+1}\| \leq |\lambda_j - 1|\|V\chi_j^t\| + |\lambda_j|\|(\widetilde{V} - V)\chi_j^t\|,    
    \end{align*}
    so we aim at bounding $|\lambda_j - 1| \leq \frac{\varepsilon}{2\|V\chi_j^t\|}$ and $\|(\widetilde{V} - V)\chi_j^t\| \leq \frac{\varepsilon}{2|\lambda_j|}$. Let us start with $|\lambda_j-1|$. Notice that $|P_j|$ is a binomial random variable with mean $p|\mathcal{C}_j^t|/n$. By a Chernoff bound (\Cref{fact:chernoff}),
    \begin{align}\label{eq:chernoffq-means}
        \operatorname{Pr}\left[\left|\frac{|P_j|}{p}\frac{n}{|\mathcal{C}_j^t|} - 1\right| \geq \frac{\varepsilon}{4\|V\|\|\chi_j^t\|}\right] \leq 2\exp\left(-\frac{\varepsilon^2 p |\mathcal{C}_j^t|}{48n\|V\|^2\|\chi_j^t\|^2}\right).
    \end{align}
    %
% \textcolor{red}{It suffices to take $p = O\big(\frac{\|V\|_F^2 n}{\varepsilon^2|\mathcal{C}_j^t|^2}\log\frac{k}{\delta}\big)$.}\\ 
    It suffices to take $p = \frac{48\|V\|^2\|\chi_j^t\|^2n}{\varepsilon^2|\mathcal{C}_j^t|}\ln\frac{2k}{\delta} =  O\big(\frac{\|V\|^2}{n} \frac{k^2}{\varepsilon^2}\log\frac{k}{\delta}\big)$ in order to estimate $|\lambda_j^{-1} - 1|\leq \frac{\varepsilon}{4\|V\|\|\chi_j^t\|}$ with probability at least $1-\frac{\delta}{2k}$ (using that $|\mathcal{C}_j^t| = \Omega(\frac{n}{k})$, $\|V\chi_j^t\| \leq \|V\|\|\chi_j^t\|$, and $\|\chi_j^t\|^2 = 1/|\mathcal{C}_j^t|$). The bound on $|\lambda_j^{-1} - 1|$ implies that $|\lambda_j - 1| \leq \frac{\varepsilon}{2\|V\|\|\chi_j^t\|}$, where we used~\Cref{lem:error_propagation} and that $|\lambda_j| \leq 2$ with high probability --- which is already implied by the bound in \Cref{eq:chernoffq-means}. By the union bound, the bound on $\lambda_j$ holds for all $j\in[k]$ with probability at least $1-\frac{\delta}{2}$.

    Regarding the bound on $\widetilde{V}$, we use Freedman's inequality to prove $\|(\widetilde{V} - V)\chi_j^t\| \leq \frac{\varepsilon}{4} \leq \frac{\varepsilon}{2|\lambda_j|}$ (again using that $|\lambda_j| \leq 2$). For such, let $f(X_1,\dots,X_q) = \|(\widetilde{V} - V)\chi_j^t\|$. First, for all $i\in[q]$,
    \begin{align*}
        |f(X_1,\dots, X_i,\dots,X_q) - f(X_1,\dots,X'_i,\dots,X_q)| 
        \leq \frac{1}{q}\|(X_i - X'_i)\chi_j^t\|
        \leq \frac{2\|V\|_{2,1}}{q|\mathcal{C}_j^t|}.
    \end{align*}
    Second, we bound the variance: for all $i\in[q]$,
    \begin{align*}
        \operatorname*{\mathbb{E}}_{X_i,X'_i}\big[f(X_1,\dots,X_i,\dots,X_q) -  f(X_1,\dots,X'_i,\dots,X_q)\big]^2
        \leq \frac{1}{q^2}\operatorname*{\mathbb{E}}_{X_i,X'_i}\big[\|(X_i - X'_i)\chi_j^t\|\big]^2 \\ 
        \leq \frac{4}{q^2}{\operatorname*{\mathbb{E}}_{X_i}}[\|X_i\chi_j^t\|^2] = \frac{4}{q^2}(\chi_j^t)^\top \operatorname*{\mathbb{E}}_{X}[X^\top X]\chi_j^t = \frac{4}{q^2}(\chi_j^t)^\top \|V\|_{2,1}\operatorname{diag}((\|v_i\|)_{i\in[n]})\chi_j^t 
        \leq \frac{4\|V\|_{2,1}^2}{q^2|\mathcal{C}_j^t|^2}.
    \end{align*}
    We employ Freedman's inequality with the Doob martingale $Y_i := \mathbb{E}[f(X_1,\dots,X_q)|X_1,\dots,X_i]$ for $i\in[q]$. Then \Cref{thm:freedman_inequality} with $B = \frac{2\|V\|_{2,1}}{q|\mathcal{C}_j^t|}$ and $\sigma^2 = q \cdot \frac{4\|V\|_{2,1}^2}{q^2|\mathcal{C}_j^t|^2}$ leads to
    \begin{align*}
        \operatorname{Pr}\left[\|(\widetilde{V} - V)\chi_j^t\| \geq \frac{\varepsilon}{4} \right] \leq \exp\left(-\frac{\varepsilon^2/32}{\frac{4\|V\|_{2,1}^2}{q|\mathcal{C}_j^t|^2} + \frac{\varepsilon\|V\|_{2,1}}{6q|\mathcal{C}_j^t|}}\right).% \leq \exp\left(-\frac{q\varepsilon^2|\mathcal{C}_j^t|}{(128 + 6\varepsilon)\|V\|_F^2}\right).
    \end{align*}
    It suffices to take $q = O\big({\max}\big\{\frac{\|V\|^2_{2,1}}{n^2}\frac{k^2}{\varepsilon^2},\frac{\|V\|_{2,1}}{n}\frac{k}{\varepsilon}\big\}\log\frac{k}{\delta}\big)$ to approximate $\|(\widetilde{V} - V)\chi_j^t\| \leq \frac{\varepsilon}{2}$ with probability at least $1-\frac{\delta}{2k}$ (already using that $|\mathcal{C}_j^t| = \Omega(\frac{n}{k})$). %By a union bound, $\|(\widetilde{V} - V)\chi_j^t\| \leq \frac{\varepsilon}{2}$ for all $j\in[k]$ with probability at least $1-\frac{\delta}{2}$. 
    All in all, we have $\|c_j^{\ast\:t+1} - c_j^{t+1}\| \leq \varepsilon$ with probability at least $1-\delta$ for all the centroids (using a union bound).

    We now turn our attention to the query and time complexities. In order to compute the clusters $\{\mathcal{C}_j^t\}_{j\in[k]}$, for each $i\in P\cup Q$ and $j\in[k]$, we exactly compute the distance $\|v_i - {c}_j^{t}\|$, which requires $O(d\log(nd))$ time: $O(d\log(nd))$ time to read a vector of $d$ components $v_i$ and $O(d)$ time to compute the distance. In total, we need access to at most $p+q$ vectors, so $O((p+q)d)$ classical queries, while the time complexity is $O((p+q)(kd + d\log(nd)))$, accounting for the $O(kd)$ cost in computing all distances $\|v_i - c_j^t\|$ between $v_i$ and the $k$ centroids stored in memory. Finally, the $k$ new centroids are obtained by summing $q$ $d$-dimensional vectors in $O(qd)$ time. In summary,
    \begin{itemize}
        \item Sampling $P,Q\subseteq[n]$ and querying the corresponding vectors $\{v_i\}_{i\in P\cup Q}$ takes $O((p+q)d)$ queries and $O((p+q)d\log(nd))$ time; 
        \item Obtaining the labels $\{\ell^t_i\}_{i\in P\cup Q}$ requires $O((p+q)kd)$ time; 
        \item Computing new centroids $\{{c}_j^{t+1}\}_{j\in[k]}$ by adding $q$ $d$-dimensional vectors takes $O(qd)$ time.
    \end{itemize}
    This means the overall query complexity is $O((p+q)d)$ and the time complexity is
    \[
        O\big((p+q)d(k + \log(nd))\big) = O\bigg(\bigg(\frac{\|V\|^2}{n} + \frac{\|V\|_{2,1}^2}{n^2}\bigg)\frac{k^2d}{\varepsilon^2}(k + \log(nd))\log\frac{k}{\delta}\bigg).    \qedhere
    \]
\end{proof}

\Cref{alg:k-means} computes the distances $\|v_i - {c}_j^t\|$, and thus the labels $\{\ell_i^t\}_{i\in P\cup Q}$, in an exact way via classical arithmetic circuits in $O((p+q)kd)$ time ($O(d)$ time for each pair $(v_i,{c}_j^t)$). Similarly, the new centroids ${c}_j^{t+1}$ are computed by adding $q$ $d$-dimensional vectors, $\sum_{i\in Q_j}\frac{v_i}{\|v_i\|}$. It is possible, however, to approximate $\|v_i - {c}_j^t\|$ via a sampling procedure, which allows to trade the dependence on $d$ with some norm of $V$. \Cref{alg:k-means2} describes how this can be performed and the next theorem analyses its query and time complexities.

\begin{algorithm}[t]
    \caption{Classical $(\varepsilon,\tau)$-$k$-means algorithm}
    \begin{algorithmic}[1]
    \Require Sampling access to data matrix $V=[v_1,\dots,v_n]\in\mathbb{R}^{d\times n}$, parameters $\delta,\varepsilon,\tau$.
%    \Ensure Vectors $c_1,\dots,c_k\in\mathbb{R}^d$ corresponding to centroids.
    \State Select $k$ initial centroids ${c}_1^0,\dots,{c}_k^0$ %and build sampling access to $C^0=[c_1^0,\dots,c_k^0]\in\mathbb{R}^{k\times d}$
    \For{$t=0$ until convergence}
        \State Sample $p = O\big(\frac{\|V\|^2}{n}\frac{k^2}{\varepsilon^2}\log\frac{k}{\delta}\big)$ rows $P\subseteq[n]$ uniformly from $[n]$
        \State Sample $q = O\big(\frac{\|V\|_{1,1}^2}{n^2}\frac{k^2}{\varepsilon^2}\log\frac{k}{\delta}\big)$ rows $Q\subseteq[n]\times [d]$ from $\mathcal{D}^{(1)}_{V}$
        \State \parbox[t]{\dimexpr\linewidth-\algorithmicindent}{For $i\in P$ and $(i,\cdot)\in Q$, compute $\ell_i^t\in \{j\in[k]: \|v_i - {c}_j^t\|^2 \leq \min_{j'\in[k]}\|v_i - {c}_{j'}^t\|^2 + \tau \}$ using \Cref{lem:classical_inner_product}}
        \State For $(j,l)\in[k]\times[d]$, let $P_j := \{i\in P|\ell_i = j\}$ and $Q_{jl} := \{(i,l')\in Q|(\ell_i^t,l') = (j,l)\}$
        \State For $(j,l)\in[k]\times[d]$, let the new centroids $({c}_j^{t+1})_l = \frac{p}{n|P_j|}\sum_{(i,l')\in Q_{jl}}\frac{\|V\|_{1,1}}{q}\operatorname{sgn}(V_{l'i})$
    \EndFor
\end{algorithmic}\label{alg:k-means2}
\end{algorithm}

\begin{theorem}[Classical $(\varepsilon,\tau)$-$k$-means algorithm]\label{thr:classical_algorithm2}
    Let $\varepsilon,\tau >0$, $\delta\in(0,1)$, and assume sampling access to $V=[v_1,\dots,v_n]\in\mathbb{R}^{d\times n}$. If all clusters satisfy $|\mathcal{C}_j^t| = \Omega(\frac{n}{k})$, then {\rm \Cref{alg:k-means2}} outputs centroids consistent with the $(\varepsilon,\tau)$-$k$-means algorithm with probability $1-\delta$. The complexities per iteration of {\rm \Cref{alg:k-means2}} are (up to $\polylog$ factors in $k$, $d$, $\frac{1}{\delta}$, $\frac{1}{\tau}$, $\frac{1}{\varepsilon}$, $\frac{\|V\|_F}{\sqrt{n}}$)
    \begin{align*}
        \textbf{Classical queries:}&~\widetilde{O}\bigg(\bigg(\frac{\|V\|^2}{n} + \frac{\|V\|_{1,1}^2}{n^2} \bigg)\frac{\|V\|_F^2\|V\|_{2,\infty}^2}{n}\frac{k^3}{\varepsilon^2\tau^2} \bigg), \\ 
        \textbf{Time:}&~\widetilde{O}\bigg(\bigg(\frac{\|V\|^2}{n} + \frac{\|V\|_{1,1}^2}{n^2} \bigg)\frac{\|V\|_F^2\|V\|_{2,\infty}^2}{n}\frac{k^3}{\varepsilon^2\tau^2}\log{n} \bigg).
    \end{align*}
\end{theorem}
\begin{proof}
    The proof is similar to \Cref{thr:classical_algorithm}. Let $\chi_j^t\in\mathbb{R}^n$ be such that $(\chi_{j}^{t})_{i} = \frac{1}{|\mathcal{C}_j^t|}$ if $i\in \mathcal{C}^t_{j}$ and $0$ if $i \not \in \mathcal{C}^t_{j}$. For $i\in[n]$, let $\ell_i^t \in \{j\in[k]:\|v_i - c_j^t\|^2 \leq \min_{j'\in[k]}\|v_i - c_{j'}^t\|^2 + \tau\}$. Again we sample $p$ indices $P\subseteq[n]$ uniformly from $[n]$ and let $\lambda_j = \frac{p}{|P_j|}\frac{|\mathcal{C}_j^t|}{n}$, where $P_j := \{i\in P | \ell_i^t = j\}$. On the other hand, we now sample $q$ indices $Q\subseteq[n]\times[d]$ from the distribution $\mathcal{D}^{(1)}_{V}(i,l) = \frac{|V_{li}|}{\|V\|_{1,1}}$. For each $(i,l)\in Q$, let $X_{li}\in\mathbb{R}^{d\times n}$ be the matrix formed by setting the $(l,i)$-th entry of $X$ to $\|V\|_{1,1}\operatorname{sgn}(V_{li})$ and the rest to zero. Define $\widetilde{V} := \frac{1}{q}\sum_{(i,l)\in Q} X_{li}$. Then $\mathbb{E}[\widetilde{V}] = V$. Let also $\overline{Q}:= \{i|(i,l)\in Q ~\text{for some}~l\in[d]\}$ for convenience.

    The outputs of the standard $k$-means and \Cref{alg:k-means} can be stated, respectively, as $c_j^{\ast\:t+1} = V\chi_j^t$ and $(c_j^{t+1})_l = \lambda_j (\widetilde{V}\chi_j^{t})_l = \frac{p}{n|P_j|}\sum_{(i,l')\in Q_{jl}}\frac{\|V\|_{1,1}}{q}\operatorname{sgn}(V_{l'i})$, where $Q_{jl} := \{(i,l')\in Q|(\ell_i^t,l') = (j,l)\}$. In order to bound $\|c_j^{\ast\:t+1} - c_j^{t+1}\|$, once again, by the triangle inequality, 
    \begin{align*}
        \|c_j^{\ast\:t+1} - c_j^{t+1}\| \leq |\lambda_j - 1|\|V\chi_j^t\| + |\lambda_j|\|(\widetilde{V} - V)\chi_j^t\|,    
    \end{align*}
    and we just need to show that $|\lambda_j - 1| \leq \frac{\varepsilon}{2\|V\chi_j^t\|}$ and $\|(\widetilde{V} - V)\chi_j^t\| \leq \frac{\varepsilon}{2|\lambda_j|}$. Exactly as in \Cref{thr:classical_algorithm}, it suffices to take $p=O\big(\frac{\|V\|^2}{n}\frac{k^2}{\varepsilon^2}\log\frac{k}{\delta}\big)$ in order to bound $|\lambda_j - 1| \leq \frac{\varepsilon}{2\|V\chi_j^t\|}$ with probability $1-\frac{\delta}{4}$.

    Regarding the bound on $\widetilde{V}$, we use Freedman's inequality to prove $\|(\widetilde{V} - V)\chi_j^t\| \leq \frac{\varepsilon}{4} \leq \frac{\varepsilon}{2|\lambda_j|}$ (using that $|\lambda_j| \leq 2$). For such, let $f(X_1,\dots,X_q) = \|(\widetilde{V} - V)\chi_j^t\|$. First, for all $i\in[q]$,
    \begin{align*}
        |f(X_1,\dots, X_i,\dots,X_q) - f(X_1,\dots,X'_i,\dots,X_q)| 
        \leq \frac{1}{q}\|(X_i - X'_i)\chi_j^t\|
        \leq \frac{2\|V\|_{1,1}}{q|\mathcal{C}_j^t|}.
    \end{align*}
    Second, we bound the variance: for all $i\in[q]$,
    \begin{align*}
        \operatorname*{\mathbb{E}}_{X_i,X'_i}\big[f(X_1,\dots,X_i,\dots,X_q) -  f(X_1,\dots,X'_i,\dots,X_q)\big]^2
        \leq \frac{1}{q^2}\operatorname*{\mathbb{E}}_{X_i,X'_i}\big[\|(X_i - X'_i)\chi_j^t\|\big]^2 \\ 
        \leq \frac{4}{q^2}{\operatorname*{\mathbb{E}}_{X_i}}[\|X_i\chi_j^t\|^2] = \frac{4}{q^2}(\chi_j^t)^\top \operatorname*{\mathbb{E}}_{X}[X^\top X]\chi_j^t = \frac{4}{q^2}(\chi_j^t)^\top \|V\|_{1,1}\operatorname{diag}((\|v_i\|_1)_{i\in[n]})\chi_j^t 
        \leq \frac{4\|V\|_{1,1}^2}{q^2|\mathcal{C}_j^t|^2}.
    \end{align*}
    We employ Freedman's inequality with the Doob martingale $Y_i := \mathbb{E}[f(X_1,\dots,X_q)|X_1,\dots,X_i]$ for $i\in[q]$. Then \Cref{thm:freedman_inequality} with $B = \frac{2\|V\|_{1,1}}{q|\mathcal{C}_j^t|}$ and $\sigma^2 = q \cdot \frac{4\|V\|_{1,1}^2}{q^2|\mathcal{C}_j^t|^2}$ leads to
    \begin{align*}
        \operatorname{Pr}\left[\|(\widetilde{V} - V)\chi_j^t\| \geq \frac{\varepsilon}{4} \right] \leq \exp\left(-\frac{\varepsilon^2/32}{\frac{4\|V\|_{1,1}^2}{q|\mathcal{C}_j^t|^2} + \frac{\varepsilon\|V\|_{1,1}}{6q|\mathcal{C}_j^t|}}\right).% \leq \exp\left(-\frac{q\varepsilon^2|\mathcal{C}_j^t|}{(128 + 6\varepsilon)\|V\|_F^2}\right).
    \end{align*}
    It suffices to take $q = O\big({\max}\big\{\frac{\|V\|^2_{1,1}}{n^2}\frac{k^2}{\varepsilon^2},\frac{\|V\|_{1,1}}{n}\frac{k}{\varepsilon}\big\}\log\frac{k}{\delta}\big)$ to approximate $\|(\widetilde{V} - V)\chi_j^t\| \leq \frac{\varepsilon}{2}$ with probability at least $1-\frac{\delta}{4k}$ (already using that $|\mathcal{C}_j^t| = \Omega(\frac{n}{k})$). %By a union bound, $\|(\widetilde{V} - V)\chi_j^t\| \leq \frac{\varepsilon}{2}$ for all $j\in[k]$ with probability at least $1-\frac{\delta}{2}$. 
    All in all, we have $\|c_j^{\ast\:t+1} - c_j^{t+1}\| \leq \varepsilon$ with probability at least $1-\frac{\delta}{2}$ for all the centroids (using a union bound).

    We now turn our attention to the query and time complexities. Another main difference to \Cref{thr:classical_algorithm} is that the clusters $\{\mathcal{C}_j^t\}_{j\in[k]}$ are computed by approximating the distances $\|v_i - c_j^t\|$ using an $\ell_2$-sampling procedure explained in \Cref{lem:classical_inner_product}. More precisely, for any $i\in[n]$ we can compute 
    \begin{align*}
        \ell_i^t\in \bigg\{j\in[k]: \|v_i - c_j^t\|^2 \leq \min_{j'\in[k]}\|v_i - c_{j'}^t\|^2 + \tau \bigg\}
    \end{align*}
    with probability $1-\frac{\delta}{4(p+q)}$ in $\widetilde{O}\big(\frac{\|V\|_F^2}{n}\frac{k\|v_i\|^2}{\tau^2}\log{n} \big)$ time and using $\widetilde{O}\big(\frac{\|V\|_F^2}{n}\frac{k\|v_i\|^2}{\tau^2} \big)$ queries, where $\widetilde{O}(\cdot)$ hides $\poly\log$ factors in $k$, $d$, $\frac{1}{\delta}$, $\frac{1}{\tau}$, $\frac{\|V\|_F}{\sqrt{n}}$. The classical query complexity in computing $\{\ell_i^t\}_{i\in P\cup\overline{Q}}$ is thus $\widetilde{O}\big(\frac{\|V\|_F^2}{n}\frac{k}{\tau^2}\sum_{i\in P\cup \overline{Q}}\|v_i\|^2 \big) = \widetilde{O}\big((p+q)\frac{\|V\|_F^2\|V\|_{2,\infty}^2}{n}\frac{k}{\tau^2}\big)$, while the time complexity has an extra $O(\log{n})$ factor.\footnote{It is possible to do slightly better than $\sum_{i\in P\cup \overline{Q}}\! \|v_i\|^2 \!\leq\! (p+q)\|V\|_{2,\infty}^2$ via concentration bounds. Let the random variable ${\operatorname{Pr}}\big[Y(i) \!=\! \frac{\|v_i\|^2}{\|V\|_{2,\infty}^2}\big] \!=\! \frac{1}{n}$ with mean $\mathbb{E}[Y] = \frac{\|V\|_F^2}{n\|V\|^2_{2,\infty}}$. By a Chernoff's bound, ${\operatorname{Pr}}\big[\!\sum_{i\in P} \!\|v_i\|^2 \!\geq\! 2p \frac{\|V\|_F^2}{n}\big] \!\leq\! {\exp}\big({-}\frac{p}{3}\frac{\|V\|_F^2}{n\|V\|_{2,\infty}^2} \big) \!\leq\! \frac{\delta}{2}$ if $p\!\geq\! 3\frac{n\|V\|_{2,\infty}^2}{\|V\|_F^2}\ln\!\frac{2}{\delta}$. On the other hand, let the random variable ${\operatorname{Pr}}\big[Y(i) \!=\! \frac{\|v_i\|^2}{\|V\|^2_{2,\infty}}\big] \!=\! \frac{\|v_i\|}{\|V\|_{2,1}}$ with mean $\mathbb{E}[Y] \!=\! \frac{1}{\|V\|_{2,1}}\sum_{i\in[n]}\!\frac{\|v_i\|^3}{\|V\|_{2,\infty}^2} \!\leq\! \frac{\|V\|_F^2}{\|V\|_{2,1}\|V\|_{2,\infty}}$. By a Chernoff's bound, ${\operatorname{Pr}}\big[\!\sum_{i\in \overline{Q}} \!\|v_i\|^2 \!\geq\! 2q \frac{\|V\|_F^2\|V\|_{2,\infty}}{\|V\|_{2,1}} \big] \!\leq\! {\exp}\big({-}\frac{q}{3}\frac{\|V\|_F^2}{\|V\|_{2,1}\|V\|_{2,\infty}} \big) \!\leq\! \frac{\delta}{2}$ if $q \!\geq\! 3\frac{\|V\|_{2,1}\|V\|_{2,\infty}}{\|V\|_F^2}\ln\!\frac{2}{\delta}$. Thus $\sum_{i\in P\cup \overline{Q}} \!\|v_i\|^2 \leq 2\big(p + q\frac{n\|V\|_{2,\infty}}{\|V\|_{2,1}} \big)\frac{\|V\|_F^2}{n}$ with probability $1-\delta$ for large enough $p$ and $q$.}
    %
    % To bound $\sum_{i\in P}\|v_i\|^2$, consider the random variable
    % %
    % \begin{align*}
    %     \forall i\in[n]: \operatorname{Pr}\left[Y(i) = \frac{\|v_i\|^2}{\|V\|_{2,\infty}^2}\right] = \frac{1}{n}, \qquad \mathbb{E}[Y] = \frac{\|V\|_F^2}{n\|V\|^2_{2,\infty}}.
    % \end{align*}
    % %
    % By Chernoff's bound (\Cref{fact:chernoff}),
    % %
    % \begin{align*}
    %     \operatorname{Pr}\left[\sum_{i\in P} \|v_i\|^2 \geq 2p \frac{\|V\|_F^2}{n}\right] = \operatorname{Pr}\left[\sum_{i\in P} \frac{\|v_i\|^2}{\|V\|_{2,\infty}^2} \geq 2p \frac{\|V\|_F^2}{n\|V\|_{2,\infty}^2}\right] \leq \exp\left(-\frac{p}{3}\frac{\|V\|_F^2}{n\|V\|_{2,\infty}^2} \right) \leq \frac{\delta}{4}.
    % \end{align*}
    % %
    % On the other hand, regarding $\sum_{i\in Q}\|v_i\|^2$, consider the random variable
    % %
    % \begin{align*}
    %     \forall i\in[n] : \operatorname{Pr}\left[Y(i) = \frac{\|v_i\|^2}{\|V\|^2_{2,\infty}}\right] = \frac{\|v_i\|}{\|V\|_{2,1}}, \quad \mathbb{E}[Y] = \frac{1}{\|V\|_{2,1}}\sum_{i\in[n]}\frac{\|v_i\|^3}{\|V\|_{2,\infty}^2} \leq \frac{\|V\|_F^2}{\|V\|_{2,1}\|V\|_{2,\infty}}.
    % \end{align*}
    % %
    % By Chernoff's bound (\Cref{fact:chernoff}),
    % %
    % \begin{align*}
    %     \operatorname{Pr}\left[\sum_{i\in Q} \|v_i\|^2 \geq 2q \frac{\|V\|_F^2\|V\|_{2,\infty}}{\|V\|_{2,1}} \right] \leq \exp\left(-\frac{q}{3}\frac{\|V\|_F^2}{\|V\|_{2,1}\|V\|_{2,\infty}} \right) \leq \frac{\delta}{4}.
    % \end{align*}
    % %
    % Hence, $\sum_{i\in P\cup Q}\|v_i\|^2 \leq 2\big(p\frac{\|V\|_F^2}{n} + q\frac{\|V\|_F^2\|V\|_{2,\infty}}{\|V\|_{2,1}} \big)$ with probability $1-\frac{\delta}{2}$.
    %
    In summary, the time and query complexities are:
    \begin{itemize}
        \item Sampling $P\subseteq[n],Q\subseteq[n]\times[d]$ takes $O(p+q)$ queries and $O((p+q)\log(nd))$ time;
        \item Computing the labels $\{\ell_i^t\}_{i\in P\cup \overline{Q}}$ takes $\widetilde{O}\big((p + q)\frac{\|V\|_F^2\|V\|_{2,\infty}^2}{n}\frac{k}{\tau^2}\big)$ classical queries and $\widetilde{O}\big((p + q )\frac{\|V\|_F^2\|V\|_{2,\infty}^2}{n}\frac{k}{\tau^2}\log{n}\big)$ time; 
        \item Querying the entries $\{V_{li}\}_{(i,l)\in Q}$ and computing all new centroids $\{{c}_j^{t+1}\}_{j\in[k]}$ takes $O(q)$ queries and $O(q\log(nd))$ time.
%        \item Registering and reading all $k$ new centroids takes $O(kd)$ time.
    \end{itemize}
    The total query complexity is thus $\widetilde{O}\big((p + q)\frac{\|V\|_F^2\|V\|_{2,\infty}^2}{n}\frac{k}{\tau^2}\big)$ and the time complexity is
    \[
        \widetilde{O}\bigg((p + q)\frac{\|V\|_F^2\|V\|_{2,\infty}^2}{n}\frac{k}{\tau^2}\log{n}\bigg) = \widetilde{O}\bigg(\bigg(\frac{\|V\|^2}{n} + \frac{\|V\|_{1,1}^2}{n^2} \bigg)\frac{\|V\|_F^2\|V\|_{2,\infty}^2}{n}\frac{k^3}{\varepsilon^2\tau^2}\log{n} \bigg).  \qedhere
    \]
\end{proof}

\begin{lemma}[Approximate classical cluster assignment]\label{lem:classical_inner_product}
    Assume classical sampling access to matrix $V = [v_1,\dots,v_n]\in\mathbb{R}^{d\times n}$. Let $\delta\in(0,1)$, $\tau>0$, and $0 < \varepsilon \leq \frac{\|V\|_F}{\sqrt{n}}$. Consider the centroid matrix $C^t = [c_1^t,\dots,c_k^t] \in\mathbb{R}^{d\times k}$ such that $\big\|c_j^{t} - |\mathcal{C}_j^{t-1}|^{-1}\sum_{i\in\mathcal{C}_j^{t-1}}v_i\big\| \leq \varepsilon$ with $|\mathcal{C}_j^{t-1}| = \Omega(\frac{n}{k})$ for all $j\in[k]$. For any $i\in[n]$, there is a classical algorithm that outputs $\ell_i^t\in\{j\in[k]: \|v_i - c_j^t\|^2 \leq \min_{j'\in[k]}\|v_i - c_{j'}^t\|^2 + \tau\}$ with probability $1-\delta$ in $O\big(\frac{\|V\|_F^2}{n}\frac{k\|v_i\|^2}{\tau^2}\log\frac{k}{\delta}\log(nd) \big)$ time and using $O\big(\frac{\|V\|_F^2}{n}\frac{k\|v_i\|^2}{\tau^2}\log\frac{k}{\delta} \big)$ classical queries.
\end{lemma}
\begin{proof}
    Fix $(i,j)\in[n]\times[k]$ and consider the random variable $X^{(ij)}$ such that
    \begin{align*}
        \operatorname{Pr}\left[X^{(ij)} = \|v_i\|^2\frac{(c_j^t)_l}{(v_i)_l}\right] = \frac{(v_i)_l^2}{\|v_i\|^2} \qquad \forall l\in[d].
    \end{align*}
    We can straightforwardly calculate
    \begin{align*}
        \mathbb{E}[X^{(ij)}] &= \sum_{l\in[d]}\|v_i\|^2\frac{(c_j^t)_l}{(v_i)_l} \frac{(v_i)_l^2}{\|v_i\|^2} = \sum_{l\in[d]} (c_j^t)_l (v_i)_l =\langle c_j^t, v_i\rangle, \\
        \operatorname{Var}[X^{(ij)}] &\leq \sum_{l\in[d]}\|v_i\|^4\frac{(c_j^t)_l^2}{(v_i)_l^2} \frac{(v_i)_l^2}{\|v_i\|^2} = \|v_i\|^2\sum_{l\in[d]} (c_j^t)^2_l =  \|v_i\|^2\|c_j^t\|^2.
    \end{align*}
    By taking the median of $K = 8\ln\frac{k}{\delta}$ copies of the mean of $\frac{64}{\tau^2}\|v_i\|^2\|c_j^t\|^2$ copies of $X^{(ij)}$, we obtain an estimate of $\langle c_j^t, v_i\rangle$ within additive error $\frac{\tau}{4}$ with probability $1-\frac{\delta}{k}$ (\Cref{fact:median-of-means}). From this, we can output $w_{ij}\in\mathbb{R}$ such that $|w_{ij} - \|v_i - c_j^t\|^2| \leq \frac{\tau}{2}$ with probability $1 - \frac{\delta}{k}$. Let $\ell_i^t = \arg\min_{j\in[k]}w_{ij}$. Then $\ell_i^t\in\{j\in[k] \!:\! \|v_i - c_j^t\|^2 \leq \min_{j'\in[k]}\!\|v_i - c_{j'}^t\|^2 + \tau\}$ with probability $1-\delta$ by a union bound.

    Regarding the sample and time complexities, the total amount of samples is
    \begin{align*}
        O\Bigg(\frac{\|v_i\|^2}{\tau^2}\log\frac{k}{\delta}\sum_{j\in[k]}\|c_j^t\|^2\Bigg) = O\Bigg(\frac{\|v_i\|^2\|C^t\|^2_F}{\tau^2}\log\frac{k}{\delta} \Bigg) = O\Bigg(\frac{\|v_i\|^2}{\tau^2}\frac{k\|V\|_F^2}{n}\log\frac{k}{\delta} \Bigg),
    \end{align*}
    where we used that $\|c_j^t\| \leq \varepsilon + |\mathcal{C}_j^{t-1}|^{-1}\sum_{i\in\mathcal{C}_j^{t-1}} \|v_i\|$ implies
    \begin{align*}
        \|C^t\|_F^2 \leq 2k\varepsilon^2 + \!\sum_{j\in[k]}\frac{2}{|\mathcal{C}_j^{t-1}|^2} \Bigg(\sum_{i\in\mathcal{C}_j^{t-1}} \!\!\|v_i\| \Bigg)^2 \!\!\leq 2k\varepsilon^2 + \!\sum_{j\in[k]}\frac{2}{|\mathcal{C}_j^{t-1}|} \sum_{i\in\mathcal{C}_j^{t-1}} \!\!\|v_i\|^2 = O\left(\frac{k\|V\|_F^2}{n} \right),
    \end{align*}
    using that $|\mathcal{C}_j^{t-1}| = \Omega(\frac{n}{k})$ for all $j\in[k]$ and $\varepsilon \leq \frac{\|V\|_F}{\sqrt{n}}$. The total time complexity is simply the sample complexity times $O(\log(nd))$.
\end{proof}

\section{Quantum algorithms}

We now describe our quantum $(\varepsilon,\tau)$-$k$-means algorithm. Similarly to our classical algorithm from the previous section, we approximate the quantities $\sum_{i\in\mathcal{C}_j^t} v_i$ and $|\mathcal{C}_j^t|$ for all $j\in[k]$ separately. This time, however, we employ quantum query access from \Cref{lem:kp_tree} to build quantum unitaries which are fed into the multivariate quantum Monte Carlo subroutine from Cornelissen, Hamoudi, and Jerbi~\cite{cornelissen2022near}. As an intermediary step, the quantities $\ell_i^t = \arg\min_{j\in[k]}\|v_i - c_j^t\|$ are computed in superposition as part of these unitaries (\Cref{lem:exact_quantum_inner_product}).

\begin{theorem}[Quantum $(\varepsilon,0)$-$q$-means algorithm]\label{thr:quantum_q_means_1}
    Let $\varepsilon>0$, $\delta\in(0,1)$, and assume quantum query access to $V=[v_1,\dots,v_n]\in\mathbb{R}^{d\times n}$. If all clusters satisfy $|\mathcal{C}_j^t| = \Omega(\frac{n}{k})$, then {\rm \Cref{alg:q-means}} outputs centroids consistent with the $(\varepsilon,\tau=0)$-$k$-means with probability $1-\delta$. The complexities per iteration of {\rm \Cref{alg:q-means}} are (up to $\polylog$ factors in $k$, $d$, $\frac{1}{\delta}$, $\frac{1}{\varepsilon}$, $\frac{\|V\|_F}{\sqrt{n}}$)
    \begin{align*}
        \textbf{Quantum queries:}&~\widetilde{O}\bigg(\bigg(\sqrt{k}\frac{\|V\|}{\sqrt{n}} + \sqrt{d}\frac{\|V\|_{2,1}}{n}\bigg)\frac{k^{\frac{3}{2}}d}{\varepsilon}\bigg),\\ % = \widetilde{O}\bigg(\frac{\|V\|_F}{\sqrt{n}}\frac{k^{\frac{3}{2}}d}{\varepsilon}(\sqrt{k} + \log{n})(\sqrt{k} + \sqrt{d})\bigg),
        \textbf{Time:}&~\widetilde{O}\bigg(\bigg(\sqrt{k}\frac{\|V\|}{\sqrt{n}} + \sqrt{d}\frac{\|V\|_{2,1}}{n}\bigg)\frac{k^{\frac{3}{2}}d}{\varepsilon}(\sqrt{k} + \log{n})\bigg).
    \end{align*}
\end{theorem}

\begin{algorithm}[t]
    \caption{Quantum $(\varepsilon,\tau=0)$-$k$-means algorithm}
    \begin{algorithmic}[1]
    \Require Quantum query access to data matrix $V=[v_1,\dots,v_n]\in\mathbb{R}^{d\times n}$, parameters $\delta,\varepsilon$.
%    \Ensure Vectors $c_1,\dots,c_k\in\mathbb{R}^d$ corresponding to centroids.
    \State Select $k$ initial centroids ${c}_1^0,\dots,{c}_k^0$ 
    \For{$t=0$ until convergence}
        \State Build quantum query access to $[{c}_1^t,\dots,{c}_k^t]\in\mathbb{R}^{d\times k}$
        \State \parbox[t]{\dimexpr\linewidth-\algorithmicindent}{Using \Cref{lem:exact_quantum_inner_product} to obtain $|i\rangle|\bar{0}\rangle \mapsto |i\rangle|\ell_i^t\rangle$ where $\ell_i^t = \arg\min_{j\in[k]}\|v_i - c_j^t\|$, construct the unitaries
        \begin{align*}
            U_{I}&:|\bar{0}\rangle \mapsto \sum_{i\in[n]}\frac{1}{\sqrt{n}}|i,\ell_i^t\rangle, \qquad\qquad 
            U_{V}:|\bar{0}\rangle \mapsto \sum_{i\in[n]} \sqrt{\frac{\|v_i\|}{\|V\|_{2,1}}} |i,\ell_i^t\rangle, \\
            \mathcal{B}_I &: |i,j\rangle|\bar{0}\rangle^{\otimes k} \mapsto |i,j\rangle|\bar{0}\rangle^{\otimes (j-1)}|1\rangle|\bar{0}\rangle^{\otimes (k-j-1)},\\
            \mathcal{B}_V &: |i,j\rangle|\bar{0}\rangle^{\otimes kd} \mapsto |i,j\rangle|\bar{0}\rangle^{\otimes (j-1)d}|v_i/\|v_i\|\rangle|\bar{0}\rangle^{\otimes (k-j-1)d}
        \end{align*}}

        \State  \parbox[t]{\dimexpr\linewidth-\algorithmicindent}{Apply the multivariate quantum Monte Carlo routine (\Cref{fact:multivariate_Monte_Carlo}) with $p = \widetilde{O}\big(\frac{\|V\|}{\sqrt{n}}\frac{k^{3/2}}{\varepsilon}\big)$ queries to the unitaries $U_{I}$ and $\mathcal{B}_I$ to obtain $P\in\mathbb{R}^{k}$}
        \State  \parbox[t]{\dimexpr\linewidth-\algorithmicindent}{Apply the multivariate quantum Monte Carlo routine (\Cref{fact:multivariate_Monte_Carlo}) with $q = \widetilde{O}\big(\frac{\|V\|_{2,1}}{\sqrt{n}}\frac{k\sqrt{d}}{\varepsilon}\big)$ queries to the unitaries $U_{V}$ and $\mathcal{B}_V$ to obtain $Q\in(\mathbb{R}^{d})^k$}
        \State For $j\in[k]$, record the new centroids ${c}_j^{t+1} = \frac{\|V\|_{2,1}}{n}\frac{Q_j}{P_j}$
    \EndFor
\end{algorithmic}\label{alg:q-means}
\end{algorithm}
\begin{proof}
    We start with the error analysis. Consider the unitaries
    \begin{align*}
        &U_{V}:|\bar{0}\rangle \mapsto \sum_{i\in[n]}\sqrt{\frac{\|v_i\|}{\|V\|_{2,1}}} |i,\ell_i^t\rangle, \quad &&\mathcal{B}_V : |i,j\rangle|\bar{0}\rangle^{\otimes kd} \mapsto |i,j\rangle|\bar{0}\rangle^{\otimes (j-1)d}|v_i/\|v_i\|\rangle|\bar{0}\rangle^{\otimes (k-j-1)d},\\
        &U_{I}:|\bar{0}\rangle \mapsto \sum_{i\in[n]}\frac{1}{\sqrt{n}}|i,\ell_i^t\rangle, \quad &&\mathcal{B}_I : |i,j\rangle|\bar{0}\rangle^{\otimes k} \mapsto |i,j\rangle|\bar{0}\rangle^{\otimes (j-1)}|1\rangle|\bar{0}\rangle^{\otimes (k-j-1)}.
    \end{align*}
    The unitaries $U_{V}$ and $U_{I}$ can be thought of as preparing a superposition over probability spaces with distributions $\mathbb{P}_V$ and $\mathbb{P}_I$, respectively, given by
    \begin{align*}
        \mathbb{P}_V(i,j) = \begin{cases}
            \frac{\|v_i\|}{\|V\|_{2,1}} &\text{if}~j = \ell_i^t,\\
            0 &\text{if}~j \neq \ell_i^t,
        \end{cases} \qquad\text{and}\qquad 
        \mathbb{P}_I(i,j) = \begin{cases}
            \frac{1}{n} &\text{if}~j = \ell_i^t,\\
            0 &\text{if}~j \neq \ell_i^t,
        \end{cases}
    \end{align*}
    while the unitaries $\mathcal{B}_V$ and $\mathcal{B}_I$ can be thought of as binary encoding the random variables $X_V:[n]\times[k]\to(\mathbb{R}^d)^k$ and $X_I:[n]\times[k]\to\mathbb{R}^k$, respectively, given by $X_V(i,j) = (0,\dots,0,\frac{v_i}{\|v_i\|},0,\dots,0)$ and $X_I(i,j) = (0,\dots,0,1,0,\dots,0)$, where the non-zero entry is the $j$-th entry. Note that 
    \begin{align*}
        \sum_{(i,j)\in[n]\times[k]} \mathbb{P}_V(i,j) X_V(i,j) &= \Bigg(\sum_{i\in \mathcal{C}_1^t} \frac{v_i}{\|V\|_{2,1}}, \dots, \sum_{i\in \mathcal{C}_k^t} \frac{v_i}{\|V\|_{2,1}}\Bigg),\\
        \sum_{(i,j)\in[n]\times[k]} \mathbb{P}_I(i,j) X_I(i,j) &= \left(\frac{|\mathcal{C}_1^t|}{n}, \dots, \frac{|\mathcal{C}_k^t|}{n}\right).
    \end{align*}
    Therefore, the multivariate quantum Monte Carlo subroutine (\Cref{fact:multivariate_Monte_Carlo}) returns $P\in\mathbb{R}^k$ and $Q\in(\mathbb{R}^d)^k$ such that, with probability at least $1-\delta$ and for some $\varepsilon_1,\varepsilon_2 >0$ to be determined,
    \begin{align*}
        \left|P_j - \frac{|\mathcal{C}_j^t|}{n} \right| \leq \varepsilon_1 \qquad\text{and}\qquad \Bigg\|Q_j - \sum_{i\in\mathcal{C}_j^t} \frac{v_i}{\|V\|_{2,1}} \Bigg\|_\infty \leq \varepsilon_2 \qquad\forall j\in[k].
    \end{align*}
    This means that, by a triangle inequality,
    \begin{align*}
        \| c_j^{\ast\:t+1} - c_j^{t+1} \| \leq \left|\frac{|\mathcal{C}_j^t|}{nP_j} - 1\right| \Bigg\|\frac{1}{|\mathcal{C}_j^t|}\sum_{i\in\mathcal{C}_j^t}v_i\Bigg\| + \frac{\|V\|_{2,1}}{nP_j}\Bigg\| Q_j - \sum_{i\in\mathcal{C}_j^t} \frac{v_i}{\|V\|_{2,1}} \Bigg\|.
    \end{align*}
    For $\varepsilon_1$ small enough such that $\varepsilon_1 \leq \min_{j\in[k]}\frac{|\mathcal{C}_j^t|}{2n}$, then $\big|P_j - \frac{|\mathcal{C}_j^t|}{n}\big| \leq \varepsilon_1 \implies \frac{1}{nP_j} \leq \frac{1}{|\mathcal{C}_j^t| - n\varepsilon_1} \leq \frac{2}{|\mathcal{C}_j^t|}$ and $\big|\frac{1}{P_j} - \frac{n}{|\mathcal{C}_j^t|}\big| \leq \frac{2n^2}{|\mathcal{C}_j^t|^2}\varepsilon_1$ according to \Cref{lem:error_propagation}. 
    Moreover, $\big\|\frac{1}{|\mathcal{C}_j^t|}\sum_{i\in\mathcal{C}_j^t}v_i\big\| = \|V \chi_j^t\| \leq \|V\|\|\chi_j^t\| = \|V\|/\sqrt{|\mathcal{C}_j^t|}$. Hence
    \begin{align*}
        \| c_j^{\ast\:t+1} - c_j^{t+1} \| \leq \frac{\|V\|}{\sqrt{|\mathcal{C}_j^t|}} \frac{2n}{|\mathcal{C}_j^t|}\varepsilon_1 + \frac{2\sqrt{d}\|V\|_{2,1}}{|\mathcal{C}_j^t|}\varepsilon_2.
    \end{align*}
    It suffices to take $\varepsilon_1 = O\big(\frac{\sqrt{n}}{\|V\|}\frac{\varepsilon}{k^{3/2}}\big)$ and $\varepsilon_2 = O\big(\frac{n}{\|V\|_{2,1}}\frac{\varepsilon}{k\sqrt{d}} \big)$ in order to obtain $\| c_j^{\ast\:t+1} - c_j^{t+1} \| \leq \varepsilon$, where we used that $|\mathcal{C}_j^t| = \Omega(\frac{n}{k})$. In order to obtain $\varepsilon_1 = O\big(\frac{\sqrt{n}}{\|V\|}\frac{\varepsilon}{k^{3/2}}\big)$, we must query the unitaries $U_I$ and $\mathcal{B}_I$ in the multivariate quantum Monte Carlo subroutine $p = \widetilde{O}\big(\frac{\|V\|}{\sqrt{n}}\frac{k^{3/2}}{\varepsilon}\big)$ times (since $\|X_I\| = 1$ and $\mathbb{E}[\|X_I\|] = 1$). On the other hand, in order to obtain $\varepsilon_2 = O\big(\frac{n}{\|V\|_{2,1}}\frac{\varepsilon}{k\sqrt{d}} \big)$, we must query the unitaries $U_V$ and $\mathcal{B}_V$ in the multivariate quantum Monte Carlo subroutine $q = \widetilde{O}\big(\frac{\|V\|_{2,1}}{n}\frac{k\sqrt{d}}{\varepsilon}\big)$ times (since $\|X_V\| = 1$ and $\mathbb{E}[\|X_V\|] = 1$).
    
    Finally, we must show how to perform the unitaries $U_V$, $U_I$, $\mathcal{B}_V$, $\mathcal{B}_I$. The binary-encoding unitary $\mathcal{B}_V$ is $d$ QRAM calls ($O(d\log{n})$ time), followed by a normalisation computation ($O(d)$ time), followed by $d$ controlled-SWAPs on $k$ qubits ($O(kd\log{k})$ time~\cite{simulating2015berry}), while $\mathcal{B}_I$ is simply $1$ controlled-SWAP on $k$ qubits. On the other hand, the probability-distribution-encoding unitaries $U_V$, $U_I$ can be performed via the initial state preparations $|\bar{0}\rangle \mapsto \sum_{i\in[n]} \sqrt{\frac{\|v_i\|}{\|V\|_{2,1}}} |i\rangle$ and $|\bar{0}\rangle \mapsto \frac{1}{\sqrt{n}}\sum_{i\in[n]}|i\rangle$, respectively, followed by the mapping $|i\rangle|\bar{0}\rangle \mapsto |i\rangle|\ell_i^t\rangle$. In \Cref{lem:exact_quantum_inner_product} we show how to implement the mapping $|i\rangle|\bar{0}\rangle \mapsto |i\rangle|\ell_i^t\rangle$ in $O(\sqrt{k}d\log\frac{1}{\delta}\log{n})$ time using $O(\sqrt{k}d\log\frac{1}{\delta})$ quantum queries. In summary,
    \begin{enumerate}
        \item $\mathcal{B}_V$ requires $O(d)$ quantum queries and $O(d\log{n} + kd\log{k})$ time;\footnote{If one has access to a quantum random access gate (QRAG)~\cite{ambainis2007quantum,Allcock2024constantdepth}, which is the unitary that performs $|i\rangle|b\rangle|x_1,\dots,x_N\rangle \mapsto |i\rangle|x_i\rangle|x_1,\dots,x_{i-1},b,x_{i+1},\dots,x_N\rangle$ in $O(\log{N})$ time for $i\in[N]$, then $\mathcal{B}_V$ requires only $O(d\log{n})$ time and the final runtime of \Cref{alg:q-means} becomes $\widetilde{O}\big(\big(\sqrt{k}\frac{\|V\|}{\sqrt{n}} + \sqrt{d}\frac{\|V\|_{2,1}}{n}\big)\frac{k^{3/2}d}{\varepsilon}\log{n}\big)$.\label{footnote3}}
        \item $U_V$ requires $O(\sqrt{k}d\log\frac{1}{\delta})$ quantum queries and $O(\sqrt{k}d\log\frac{1}{\delta}\log{n})$ time;
        \item $\mathcal{B}_I$ requires no quantum queries and $O(k\log{k})$ time;
        \item $U_I$ requires $O(\sqrt{k}d\log\frac{1}{\delta})$ quantum queries and $O(\sqrt{k}d\log\frac{1}{\delta}\log{n})$ time.
    \end{enumerate}
    Collecting all the terms, the total number of quantum queries is $\widetilde{O}((p+q)\sqrt{k}d)$, while the overall time complexity is
    \[
        \widetilde{O}\left((p+q)\sqrt{k}d(\sqrt{k} + \log{n})\right) = \widetilde{O}\bigg(\bigg(\sqrt{k}\frac{\|V\|}{\sqrt{n}} + \sqrt{d}\frac{\|V\|_{2,1}}{n}\bigg)\frac{k^{\frac{3}{2}}d}{\varepsilon}(\sqrt{k} + \log{n})\bigg).  \qedhere
    \]
\end{proof}

\begin{lemma}[Exact quantum cluster assignment]
    \label{lem:exact_quantum_inner_product}
    Let $\delta\in(0,1)$ and assume quantum query access to matrices $V=[v_1,\dots,v_n]\in\mathbb{R}^{d\times n}$ and $[c_1^t,\dots,c_k^t] \in \mathbb{R}^{d\times k}$.
    % be stored in the data structure from Lemma~{\rm \ref{lem:kp_tree}}. %Let be stored in a QRAM such that the unitary operation  $|j\rangle|0\rangle\mapsto|j\rangle|c_j\rangle$ can be performed in $O(\log^2(kd))$ time. 
    There is a quantum algorithm that performs the mapping $|i\rangle|\bar{0}\rangle \mapsto |i\rangle|L^t_i\rangle$ using $O(\sqrt{k}d\log\frac{1}{\delta})$ quantum queries and in $O(\sqrt{k}d\log\frac{1}{\delta}\log{n})$ time such that, upon measuring $|L^t_i\rangle$ on the computational basis, the outcome equals $\arg\min_{j\in[k]}\|v_i-c_j^t\|$ with probability at least $1-\delta$.
\end{lemma}
\begin{proof}
    First we describe how to perform the mapping $|i,j\rangle|\bar{0}\rangle \mapsto |i,j\rangle|\|v_i - c_j^t\|\rangle$. Starting from $|i,j\rangle|\bar{0},\bar{0}\rangle|\bar{0}\rangle$, we query $2d$ times the QRAM oracles used in \Cref{lem:kp_tree} to map
    \begin{align*}
        |i,j\rangle|\bar{0},\bar{0}\rangle|\bar{0}\rangle \mapsto |i,j\rangle |v_i,c_j^t\rangle|\bar{0}\rangle.
    \end{align*}
    This operation is followed by computing the distance $\|v_i-c_j^t\|$ between the vectors $v_i$ and $c_j^t$ in $O(d)$ size and $O(\log d)$ depth by using a classical circuit, which leads to
    \begin{align*}
        |i,j\rangle |v_i,c_j^t\rangle|\bar{0}\rangle \mapsto |i,j\rangle |v_i,c_j^t\rangle|\|v_i-c_j^t\|\rangle.
    \end{align*}
    Uncomputing the first step leads to the desire state. Overall, the map $|i,j\rangle|\bar{0}\rangle \mapsto |i,j\rangle|\|v_i - c_j^t\|\rangle$ uses $O(d)$ queries to the matrices $V$ and $[c_1^t,\dots,c_k^t]$.
    
    Fix $i\in[n]$. The mapping $|j\rangle|\bar{0}\rangle \mapsto |j\rangle|\|v_i - c_j^t\|\rangle$ can be viewed as quantum access to the vector $(\|v_i - c_j^t\|)_{j\in[k]}$. Therefore, we can assign a cluster $\ell^t_i := \arg\min_{j\in[k]}\|v_i - c_j^t\|$ to the vector $v_i$ by using (controlled on $|i\rangle$) quantum minimum finding subroutine (\Cref{fact:minimum_finding}), which leads to the map $|i\rangle|\bar{0}\rangle \mapsto |i\rangle|L^t_i\rangle$, where upon measuring $|L^t_i\rangle$ on the computational basis, the outcome equals $\ell^t_i = \arg\min_{j\in[k]}\|v_i - c_j^t\|$ with probability at least $1-\delta$.
    The time cost of finding the minimum is $O(\sqrt{k}\log\frac{1}{\delta})$ queries to the unitary performing the mapping $|i,j\rangle|\bar{0}\rangle \mapsto |i,j\rangle|\|v_i-c_j^t\|\rangle$, to a total time complexity of $O(\sqrt{k}d\log\frac{1}{\delta}\log{n})$ and $O(\sqrt{k}d\log\frac{1}{\delta})$ quantum queries.
\end{proof}

\begin{remark}
    It is possible to avoid QRAM access to the centroids $[c_1^t,\dots,c_k^t]\in\mathbb{R}^{d\times k}$ by accessing them through the fixed registers $\bigotimes_{j\in[k]}|c_j^t\rangle$. This, however, hinders the use of quantum minimum finding. The index $\ell_i^t = \arg\min_{j\in[k]}\|v_i - c_j^t\|$ can be found, instead, through a classical circuit on the registers $\bigotimes_{j\in[k]}|\|v_i - c_j^t\|\rangle$, which modifies the time complexity of {\rm \Cref{lem:exact_quantum_inner_product}} to $O(d(k + \log{n}))$.
\end{remark}

Similarly to classical $(\varepsilon,\tau)$-$k$-means algorithm, it is possible to approximate the distances $\|v_i - c_j^t\|$ within quantum minimum finding using inherently quantum subroutines (\Cref{lem:quantum_approximate_inner_product}) instead of a classical arithmetic circuit as in \Cref{alg:q-means}. This allows us to replace the $O(d)$ time overhead with some norm of $V$. \Cref{alg:q-means2} describes how this can be performed and the next theorem analyses its query and time complexities.
\begin{algorithm}[t]
    \caption{Quantum $(\varepsilon,\tau)$-$k$-means algorithm}
    \begin{algorithmic}[1]
    \Require Quantum query access to data matrix $V=[v_1,\dots,v_n]\in\mathbb{R}^{d\times n}$, parameters $\delta,\varepsilon,\tau$.
%    \Ensure Vectors $c_1,\dots,c_k\in\mathbb{R}^d$ corresponding to centroids.
    \State Select $k$ initial centroids ${c}_1^0,\dots,{c}_k^0$ 
    \For{$t=0$ until convergence}
        \State Build quantum query access to $[{c}_1^t,\dots,{c}_k^t]\in\mathbb{R}^{d\times k}$
        \State \parbox[t]{\dimexpr\linewidth-\algorithmicindent}{Using  \Cref{lem:quantum_approximate_inner_product} to obtain $|i\rangle|\bar{0}\rangle \mapsto |i\rangle|\ell_i^t\rangle$ such that $\ell_i^t \in \{j\in[k]: \|v_i - c_j^t\|^2 \leq \min_{j'\in[k]}\|v_i - c_{j'}^t\|^2 + \tau\}$, construct the unitaries
        \begin{align*}
            U_{I}&:|\bar{0}\rangle \mapsto \sum_{i\in[n]}\frac{1}{\sqrt{n}}|i,\ell_i^t\rangle, \qquad\qquad U_{V}:|\bar{0}\rangle \mapsto \sum_{(i,l)\in[n]\times[d]} \sqrt{\frac{|V_{il}|}{\|V\|_{1,1}}} |i,\ell_i^t,l\rangle, \\
            \mathcal{B}_I &: |i,j\rangle|\bar{0}\rangle^{\otimes k} \mapsto |i,j\rangle|\bar{0}\rangle^{\otimes (j-1)}|1\rangle|\bar{0}\rangle^{\otimes (k-j)},\\
            \mathcal{B}_V &: |i,j,l\rangle|0\rangle^{\otimes kd} \mapsto |i,j,l\rangle|0\rangle^{\otimes ((j-1)d + (l-1))}|\operatorname{sgn}(V_{li})\rangle|0\rangle^{\otimes ((k-j-1)d + (d-l-1))}
        \end{align*}}

        \State  \parbox[t]{\dimexpr\linewidth-\algorithmicindent}{Apply the multivariate quantum Monte Carlo routine (\Cref{fact:multivariate_Monte_Carlo}) with $p = \widetilde{O}\big(\frac{\|V\|}{\sqrt{n}}\frac{k^{3/2}}{\varepsilon}\big)$ queries to the unitaries $U_{I}$ and $\mathcal{B}_I$ to obtain $P\in\mathbb{R}^{k}$}
        \State  \parbox[t]{\dimexpr\linewidth-\algorithmicindent}{Apply the multivariate quantum Monte Carlo routine (\Cref{fact:multivariate_Monte_Carlo}) with $q = \widetilde{O}\big(\frac{\|V\|_{1,1}}{\sqrt{n}}\frac{k\sqrt{d}}{\varepsilon}\big)$ queries to the unitaries $U_{V}$ and $\mathcal{B}_V$ to obtain $Q\in(\mathbb{R}^{d})^k$}
        \State For $j\in[k]$, record the new centroids ${c}_j^{t+1} = \frac{\|V\|_{1,1}}{n}\frac{Q_j}{P_j}$
    \EndFor
\end{algorithmic}\label{alg:q-means2}
\end{algorithm}

\begin{theorem}[Quantum $(\varepsilon,\tau)$-$q$-means algorithm]
    Let $\varepsilon,\tau>0$, $\delta\in(0,1)$, and assume quantum query access to $V=[v_1,\dots,v_n]\in\mathbb{R}^{d\times n}$. If all clusters satisfy $|\mathcal{C}_j^t| = \Omega(\frac{n}{k})$, then {\rm \Cref{alg:q-means2}} outputs centroids consistent with the $(\varepsilon,\tau)$-$k$-means with probability $1-\delta$. The complexities per iteration of {\rm \Cref{alg:q-means2}} are (up to $\polylog$ factors in $k$, $d$, $\frac{1}{\delta}$, $\frac{1}{\tau}$, $\frac{1}{\varepsilon}$, $\frac{\|V\|_F}{\sqrt{n}}$)
    \begin{align*}
        \textbf{Quantum queries:}&~\widetilde{O}\bigg(\bigg(\sqrt{k}\frac{\|V\|}{\sqrt{n}} + \sqrt{d}\frac{\|V\|_{1,1}}{n} \bigg)\frac{\|V\|_F\|V\|_{2,\infty}}{\sqrt{n}}\frac{k^{\frac{3}{2}}}{\varepsilon\tau} \bigg),\\ % = \widetilde{O}\bigg(\frac{\|V\|_F}{\sqrt{n}}\frac{k^{\frac{3}{2}}d}{\varepsilon}(\sqrt{k} + \log{n})(\sqrt{k} + \sqrt{d})\bigg),
        \textbf{Time:}&~\widetilde{O}\bigg(\bigg(\sqrt{k}\frac{\|V\|}{\sqrt{n}} + \sqrt{d}\frac{\|V\|_{1,1}}{n} \bigg)\bigg(\frac{\|V\|_F\|V\|_{2,\infty}}{\sqrt{n}}\frac{k^{\frac{3}{2}}}{\varepsilon\tau}\log{n} + \frac{k^2d}{\varepsilon}\bigg) \bigg).
    \end{align*}
\end{theorem}
\begin{proof}
    The proof is similar to \Cref{thr:quantum_q_means_1}. The unitaries $U_{I}$ and $\mathcal{B}_I$ are still the same,
    \begin{align*}
         U_{I}:|\bar{0}\rangle \mapsto \sum_{i\in[n]}\frac{1}{\sqrt{n}}|i,\ell_i^t\rangle, \qquad \mathcal{B}_I : |i,j\rangle|0\rangle^{\otimes k} \mapsto |i,j\rangle|0\rangle^{\otimes (j-1)}|1\rangle|0\rangle^{\otimes (k-j-1)},
    \end{align*}
    but the unitaries $U_{V}$ and $\mathcal{B}_V$ are now replaced with
    \begin{align*}
        U_{V} &: |\bar{0}\rangle \mapsto \sum_{(i,l)\in[n]\times[d]}\sqrt{\frac{|V_{li}|}{\|V\|_{1,1}}} |i,\ell_i^t,l\rangle, \\
        \mathcal{B}_V &: |i,j,l\rangle|0\rangle^{\otimes kd} \mapsto |i,j,l\rangle|0\rangle^{\otimes ((j-1)d + (l-1))}|\operatorname{sgn}(V_{li})\rangle|0\rangle^{\otimes ((k-j-1)d + (d-l-1))}.
    \end{align*}
    The new unitary $U_V$ can be thought of as preparing a superposition over the probability space with distribution $\mathbb{P}_V$ given by
    \begin{align*}
        \mathbb{P}_V(i,j,l) = \begin{cases}
            \frac{|V_{li}|}{\|V\|_{1,1}} &\text{if}~j = \ell_i^t,\\
            0 &\text{if}~j \neq \ell_i^t,
        \end{cases}
    \end{align*}
    while the unitary $\mathcal{B}_V$ can be thought of as binary encoding the random variables $X_V:[n]\times[k]\times[d]\to(\mathbb{R}^d)^k$ given by $X_V(i,j,l) = (0,\dots,0,\operatorname{sgn}(V_{li}),0,\dots,0)$, where the non-zero entry is the $((j-1)d + l)$-th entry. Note that 
    \begin{align*}
        \sum_{(i,j,l)\in[n]\times[k]\times[d]} \mathbb{P}_V(i,j,l) X_V(i,j,l) &= \Bigg(\sum_{i\in \mathcal{C}_1^t} \frac{v_i}{\|V\|_{1,1}}, \dots, \sum_{i\in \mathcal{C}_k^t} \frac{v_i}{\|V\|_{1,1}}\Bigg).
    \end{align*}
    Therefore, the multivariate quantum Monte Carlo subroutine (\Cref{fact:multivariate_Monte_Carlo}) returns $P\in\mathbb{R}^k$ and $Q\in(\mathbb{R}^d)^k$ such that, with probability at least $1-\delta$ and for some $\varepsilon_1,\varepsilon_2 >0$ to be determined,
    \begin{align*}
        \left|P_j - \frac{|\mathcal{C}_j^t|}{n} \right| \leq \varepsilon_1 \qquad\text{and}\qquad \Bigg\|Q_j - \sum_{i\in\mathcal{C}_j^t} \frac{v_i}{\|V\|_{1,1}} \Bigg\|_\infty \leq \varepsilon_2 \qquad\forall j\in[k].
    \end{align*}
    Similar to \Cref{thr:quantum_q_means_1}, by a triangle inequality,
    \begin{align*}
        \| c_j^{\ast\:t+1} - c_j^{t+1} \| \leq \frac{\|V\|}{\sqrt{|\mathcal{C}_j^t|}} \frac{2n}{|\mathcal{C}_j^t|}\varepsilon_1 + \frac{2\sqrt{d}\|V\|_{1,1}}{|\mathcal{C}_j^t|}\varepsilon_2.
    \end{align*}
    It suffices to query the unitaries $U_I$ and $\mathcal{B}_I$ a number of $p = \widetilde{O}\big(\frac{\|V\|}{\sqrt{n}}\frac{k^{3/2}}{\varepsilon}\big)$ times within quantum multivariate Monte Carlo subroutine to get $\varepsilon_1 = O\big(\frac{\sqrt{n}}{\|V\|}\frac{\varepsilon}{k^{3/2}}\big)$. By the same toke, it suffices to query the unitaries $U_V$ and $\mathcal{B}_V$ a number of $q = \widetilde{O}\big(\frac{\|V\|_{1,1}}{n}\frac{k\sqrt{d}}{\varepsilon}\big)$ times within quantum multivariate Monte Carlo subroutine to get $\varepsilon_2 = O\big(\frac{n}{\|V\|_{1,1}}\frac{\varepsilon}{k\sqrt{d}} \big)$. This yields $\| c_j^{\ast\:t+1} - c_j^{t+1} \| \leq \varepsilon$ as wanted.

    We now show how to perform the unitaries $U_V$, $U_I$, $\mathcal{B}_V$, $\mathcal{B}_I$. The binary-encoding unitary $\mathcal{B}_V$ is $1$ quantum query ($O(\log{n})$ time), followed by $1$ controlled-SWAP on $kd$ qubits ($O(kd\log(kd))$ time~\cite{simulating2015berry}), while $\mathcal{B}_I$ is simply $1$ controlled-SWAP on $k$ qubits. On the other hand, the probability-distribution-encoding unitaries $U_V$, $U_I$ can be performed via the initial state preparations $|\bar{0}\rangle \mapsto \sum_{(i,l)\in[n]\times[d]} \sqrt{\frac{|V_{li}|}{\|V\|_{1,1}}} |i,l\rangle$ and $|\bar{0}\rangle \mapsto \frac{1}{\sqrt{n}}\sum_{i\in[n]}|i\rangle$, respectively, followed by the mapping $|i\rangle|\bar{0}\rangle \mapsto |i\rangle|\ell_i^t\rangle$. In \Cref{lem:quantum_approximate_inner_product} we show how to implement the mapping $|i\rangle|\bar{0}\rangle \mapsto |i\rangle|\ell_i^t\rangle$, where $\ell_i^t \in \{j\in[k]:\|v_i - c_j^t\|^2 \leq \min_{j'\in[k]}\|v_i - c^t_{j'}\|^2 + \tau\}$, using $\widetilde{O}\big(\frac{\|V\|_F}{\sqrt{n}}\frac{\sqrt{k}\|V\|_{2,\infty}}{\tau}\big)$ quantum queries and $\widetilde{O}\big(\frac{\|V\|_F}{\sqrt{n}}\frac{\sqrt{k}\|V\|_{2,\infty}}{\tau}\log{n}\big)$ time. In summary,
    \begin{enumerate}
        \item $\mathcal{B}_V$ requires $O(1)$ quantum queries and $O(\log(nd) + kd\log(kd))$ time;\footnote{If one has access to a QRAG, then $\mathcal{B}_V$ requires only $O(\log(nd))$ time and the final runtime of \Cref{alg:q-means2} is $\widetilde{O}\big(\big(\sqrt{k}\frac{\|V\|}{\sqrt{n}} + \sqrt{d}\frac{\|V\|_{1,1}}{n} \big)\frac{\|V\|_F\|V\|_{2,\infty}}{\sqrt{n}}\frac{k^{3/2}}{\varepsilon\tau}\log{n} + kd\big)$, where the term $\widetilde{O}(kd)$ comes from \Cref{footnote1}.\label{footnote4}}
        \item $U_V$ requires $\widetilde{O}\big(\frac{\|V\|_F}{\sqrt{n}}\frac{\sqrt{k}\|V\|_{2,\infty}}{\tau}\big)$ quantum queries and $\widetilde{O}\big(\frac{\|V\|_F}{\sqrt{n}}\frac{\sqrt{k}\|V\|_{2,\infty}}{\tau}\log{n}\big)$ time;
        \item $\mathcal{B}_I$ requires no quantum queries and $O(k\log{k})$ time;
        \item $U_I$ requires $\widetilde{O}\big(\frac{\|V\|_F}{\sqrt{n}}\frac{\sqrt{k}\|V\|_{2,\infty}}{\tau}\big)$ quantum queries and $\widetilde{O}\big(\frac{\|V\|_F}{\sqrt{n}}\frac{\sqrt{k}\|V\|_{2,\infty}}{\tau}\log{n}\big)$ time.
    \end{enumerate}
    Collecting all the term, the total number of quantum queries is $\widetilde{O}\big((p+q)\frac{\|V\|_F}{\sqrt{n}}\frac{\sqrt{k}\|V\|_{2,\infty}}{\tau}\big)$, while the overall time complexity is
    \begin{align*}
        \widetilde{O}&\bigg((p+q)\bigg(\frac{\|V\|_F}{\sqrt{n}}\frac{\sqrt{k}\|V\|_{2,\infty}}{\tau}\log{n} + kd\bigg) \bigg) \\
        &= \widetilde{O}\bigg(\bigg(\sqrt{k}\frac{\|V\|}{\sqrt{n}} + \sqrt{d}\frac{\|V\|_{1,1}}{n} \bigg)\bigg(\frac{\|V\|_F\|V\|_{2,\infty}}{\sqrt{n}}\frac{k^{\frac{3}{2}}}{\varepsilon\tau}\log{n} + \frac{k^2d}{\varepsilon}\bigg) \bigg). \qedhere
    \end{align*}
\end{proof}

\begin{lemma}[Approximate quantum cluster assignment]
    \label{lem:quantum_approximate_inner_product}
    Assume quantum query access to matrix $V=[v_1,\dots,v_n]\in\mathbb{R}^{d\times n}$. Let $\delta\in(0,1)$, $\tau>0$, and $0 < \varepsilon \leq \frac{\|V\|_F}{\sqrt{n}}$. Assume quantum query access to centroid matrix $C^t = [c_1^t,\dots,c_k^t] \in \mathbb{R}^{d\times k}$ such that $\|c_j^t - |\mathcal{C}_j^{t-1}|^{-1}\sum_{i\in\mathcal{C}_j^{t-1}} v_i\| \leq \varepsilon$ with $|\mathcal{C}_j^{t-1}| = \Omega(\frac{n}{k})$ for all $j\in[k]$. There is a quantum algorithm that performs the mapping $|i\rangle|\bar{0}\rangle \mapsto |i\rangle|L^t_i\rangle$ such that, upon measuring $|L^t_i\rangle$ on the computational basis, the outcome equals $\ell_i^t \in \{j\in[k] : \|v_i - c_j^t\|^2 \leq \min_{j'\in[k]}\|v_i-c_{j'}^t\|^2 + \tau\}$ with probability at least $1-\delta$. It uses $\widetilde{O}\big(\frac{\|V\|_F}{\sqrt{n}}\frac{\sqrt{k}\|V\|_{2,\infty}}{\tau}\big)$ quantum queries and $\widetilde{O}\big(\frac{\|V\|_F}{\sqrt{n}}\frac{\sqrt{k}\|V\|_{2,\infty}}{\tau}\log{n}\big)$ time, where $\widetilde{O}(\cdot)$ hides $\polylog$ factors in $k$, $d$, $\frac{1}{\delta}$, $\frac{1}{\tau}$, $\frac{\|V\|_F}{\sqrt{n}}$.
\end{lemma}
\begin{proof}
    We first describe how to perform the map $|i,j\rangle|\bar{0}\rangle \mapsto |i,j\rangle|w_{ij}\rangle$, where $|w_{ij} - \|v_i - c_j^t\|^2| \leq \frac{\tau}{2}$ with high probability. Recall from \Cref{lem:kp_tree} that we can perform the maps
    \begin{align*}
        \mathcal{O}_{V}: |i\rangle|\bar{0}\rangle \mapsto \sum_{l\in[d]} \frac{(v_i)_l}{\|v_i\|} |i,l\rangle \qquad\text{and}\qquad \mathcal{O}_{C^t}: |j\rangle|\bar{0}\rangle \mapsto \sum_{l\in[d]} \frac{(c_j^t)_l}{\|c_j^t\|} |j,l\rangle
    \end{align*}
    in $O(\log(nd))$ time. Start then with the quantum state $|i,j\rangle\frac{|0\rangle + |1\rangle}{\sqrt{2}}|\bar{0}\rangle$ and perform the above maps controlled on the third register $\frac{|0\rangle + |1\rangle}{\sqrt{2}}$, i.e., perform $\mathcal{O}_{V}$ if the third register is $|0\rangle$ and $\mathcal{O}_{C^t}$ if it is $|1\rangle$. The final state is
    \begin{align*}
        \frac{1}{\sqrt{2}}|i,j\rangle\sum_{l\in[d]}\left(\frac{(v_i)_l}{\|v_i\|} |0,l\rangle + \frac{(c_j^t)_l}{\|c_j^t\|} |1,l\rangle \right).
    \end{align*}
    After applying a Hadamard gate onto the third register, the state becomes
    \begin{align*}
        &|i,j\rangle\sum_{l\in[d]}\left(\frac{1}{2}\left(\frac{(v_i)_l}{\|v_i\|} + \frac{(c_j^t)_l}{\|c_j^t\|}\right) |0,l\rangle + \frac{1}{2}\left(\frac{(v_i)_l}{\|v_i\|} - \frac{(c_j^t)_l}{\|c_j^t\|}\right) |1,l\rangle \right)\\
        =~ &|i,j\rangle(\sqrt{p_{ij}}|0\rangle|\psi_{ij}\rangle + \sqrt{1-p_{ij}}|1\rangle|\phi_{ij}\rangle,
    \end{align*}
    where
    \begin{align*}
        p_{ij} = \frac{1}{4}\sum_{l\in[d]} \left(\frac{(v_i)_l}{\|v_i\|} + \frac{(c_j^t)_l}{\|c_j^t\|}\right)^2 = \frac{1}{2} + \frac{\langle v_i,c_j^t \rangle}{2\|v_i\|\|c_j^t\|}
    \end{align*}
    is the probability of measuring the third register on state $|0\rangle$, and $|\psi_{ij}\rangle$ and $|\phi_{ij}\rangle$ are ``garbage'' normalised states. It is then possible to apply a standard quantum amplitude estimation subroutine~\cite{brassard2002quantum} to obtain a quantum state $|i,j\rangle|\Psi'_{ij}\rangle$ such that, upon measuring onto the computation basis, the outcome $\widetilde{p}_{ij}$ is such that $|\widetilde{p}_{ij} - p_{ij}| \leq \frac{\tau}{4\|v_i\|\|c_j^t\|} \implies |w_{ij} - \|v_i - c_j^t\|^2| \leq \frac{\tau}{2}$ with probability at $1-\delta_2$ for some $\delta_2\in(0,1)$, where $w_{ij} = \|v_i\|^2 + \|c_j^t\|^2 - \|v_i\|\|c_j^t\|(2\widetilde{p}_{ij} - 1)$. It is then straightforward to obtain a new state $|\Psi_{ij}\rangle$ from $|\Psi'_{ij}\rangle$ which returns $w_{ij}$ upon measurement with probability $1-\delta_2$. For each $(i,j)\in[n]\times[k]$, mapping $|i,j\rangle|\bar{0}\rangle \mapsto |i,j\rangle|w_{ij}\rangle$ requires $O\big(\frac{\|v_i\|\|c_j^t\|}{\tau}\log\frac{1}{\delta_2}\big)$ quantum queries and $O\big(\frac{\|v_i\|\|c_j^t\|}{\tau}\log\frac{1}{\delta_2}\log(nd)\big)$ time.

    Fix $i\in[n]$. The mapping $|i,j\rangle|\bar{0}\rangle \mapsto |i,j\rangle|w_{ij}\rangle$ can be viewed as quantum access to the vector $(w_{ij})_{j\in[k]}$. We thus employ the (variable-time) quantum minimum finding subroutine (\Cref{fact:variable_time_minimum_finding}) in order to obtain the map $|i\rangle|\bar{0}\rangle \mapsto |i\rangle|L_i^t\rangle$, where upon measuring $|L_i^t\rangle$ on the computation basis, the outcome equals $\ell_i^t = \arg\min_{j\in[k]}w_{ij} \in \{j\in[k]: \|v_i - c_j^t\|^2 \leq \min_{j'\in[k]}\|v_i - c_{j'}^t\|^2 + \tau\}$ with probability $1-\delta_1$. According to \Cref{fact:variable_time_minimum_finding}, the query complexity of $|i\rangle|\bar{0}\rangle \mapsto |i\rangle|L_i^t\rangle$ is
    \begin{align}\label{eq:classical_approximate_complexities}
        O\bigg(\frac{\|v_i\|}{\tau}\log\frac{1}{\delta_1}\log\frac{1}{\delta_2}\sqrt{\sum_{j\in[k]}\|c_j^t\|^2}  \bigg) = O\bigg(\frac{\|V\|_F}{\sqrt{n}}\frac{\sqrt{k}\|V\|_{2,\infty}}{\tau}\log\frac{1}{\delta_1}\log\frac{1}{\delta_2} \bigg),
    \end{align}
    where we used that $\sum_{j\in[k]} \|c_j^t\|^2 = \|C^t\|_F^2 = O\big(k\frac{\|V\|_F^2}{n} \big)$ as in \Cref{lem:classical_inner_product} and $\|v_i\| \leq \|V\|_{2,\infty}$, while the time complexity is $O(\log(nd))$ times the query complexity.
    
    In order to analyse the success probability (see~\cite[Appendix~A]{chen2023quantum} for a similar argument), on the other hand, first note that we implement the unitary $\widetilde{U} : |i,j\rangle|\bar{0}\rangle \mapsto |i,j\rangle(\sqrt{1-\delta_2}|w_{ij}\rangle + \sqrt{\delta_2}|w_{ij}^\perp\rangle)$, where $|w_{ij}\rangle$ contains the approximation $|w_{ij} - \|v_i - c_j^t\|^2| \leq \frac{\tau}{2}$ and $|w_{ij}^\perp\rangle$ is a normalised quantum state orthogonal to $|w_{ij}\rangle$. Ideally, we would like to implement $U : |i,j\rangle|\bar{0}\rangle \mapsto |i,j\rangle|w_{ij}\rangle$. Also
    \begin{align*}
        \forall |i,j\rangle :\qquad \|(U - \widetilde{U})|i,j\rangle|\bar{0}\rangle\| = \sqrt{(1 - \sqrt{1-\delta_2})^2 + \delta_2} = \sqrt{2-2\sqrt{1-\delta_2}} \leq \sqrt{2\delta_2},
    \end{align*}
    using that $\sqrt{1-\delta_2} \geq 1 - \delta_2$. Since (variable-time) quantum minimum finding does not take into account the action of $U$ onto states of the form $|i,j\rangle|\bar{0}^\perp\rangle$ for $|\bar{0}^\perp\rangle$ orthogonal to $|\bar{0}\rangle$, we can, without of loss of generality, assume that $\|U - \widetilde{U}\| \leq \sqrt{2\delta_2}$. The success probability of (variable-time) quantum minimum finding is $1-\delta_1$ when employing the unitary $U$. However, since it employs $\widetilde{U}$ instead, the success probability decreases by at most the spectral norm of the difference between the ``real'' and the ``ideal'' total unitaries. To be more precise, the ``ideal'' (variable-time) quantum minimum finding is a sequence of gates $\mathcal{A} = U_1E_1U_2 E_2\cdots U_{N}E_N$, where $U_i \in \{U,U^\dagger\}$, $E_i$ is a circuit of elementary gates, and $N$ is the number of queries to $U$, which can be upper-bounded as $O\big(\sqrt{k}\log\frac{1}{\delta_1}\big)$. The ``real'' implementation, on the other hand, is $\widetilde{\mathcal{A}} = \widetilde{U}_1E_1\widetilde{U}_2 E_2\cdots \widetilde{U}_{N}E_N$, where $\widetilde{U}_i \in \{\widetilde{U},\widetilde{U}^\dagger\}$. Then $\|\mathcal{A} - \widetilde{\mathcal{A}}\| \leq N\|U - \widetilde{U}\| \leq N\sqrt{2\delta_2}$. The failure probability is thus $\delta_1 + N\sqrt{2\delta_2}$. By taking $\delta_1 = O(\delta)$ and $\delta_2 = O\big(\frac{\delta^2}{N^2}\big)$, the success probability is $1-\delta$. The final complexities are obtained by replacing $\delta_1$ and $\delta_2$ into \Cref{eq:classical_approximate_complexities}.
\end{proof}

\section{Lower bounds}

In this section we prove query lower bounds to the matrix $V = [v_1,\dots,v_n]\in\mathbb{R}^{d\times n}$ for finding new centroids $c_1,\dots,c_k\in\mathbb{R}^d$ given $k$ clusters $\{\mathcal{C}_j\}_{j\in[k]}$ that form a partition of $[n]$. We note that the task considered here is easier than the one performed by $(\varepsilon,\tau)$-$k$-means, since the clusters $\{\mathcal{C}_j\}_{j\in[k]}$ are part of the input. Nonetheless, query lower bounds for such problem will prove to be tight in most parameters. The main idea is to reduce the problem of approximating $\frac{1}{|\mathcal{C}_j|}\sum_{i\in\mathcal{C}_j} v_i$ for all $j\in[k]$ from the problem of approximating the Hamming weight of some bit-string, whose query complexity is given in the following well-known fact.
\begin{fact}[{\cite[Theorem 1.11]{nayak1999quantum}}]\label{fact:lower_bound_reduction}
    Let $x\in\{0,1\}^n$ be a bit-string with Hamming weight $|x| = \Theta(n)$ accessible through queries. Consider the problem of outputting $w\in[n]$ such that $||x| - w| \leq m$ for a given $0\leq m \leq n/4$. Its randomised classical query complexity is $\Theta(\min\{(n/m)^2,n\})$, while its quantum query complexity is $\Theta(\min\{n/m, n\})$.
\end{fact}

Before presenting our query lower bounds for $(\varepsilon,\tau)$-$k$-means, we shall need the following fact.
\begin{lemma}\label{fact:simple_fact}
    Given $x\in\mathbb{R}^d$ such that $\|x\|_1 \leq \varepsilon$ for $\varepsilon \geq 0$, there is $S\subseteq[d]$ with $|S| \geq \lceil\frac{d}{2}\rceil$ such that $|x_i| \leq \frac{2\varepsilon}{d}$ for all $i\in S$.
\end{lemma}
\begin{proof}
    Arrange the entries of $x$ is descending order, i.e., $|x_{k_1}| \geq |x_{k_2}| \geq \cdots \geq |x_{k_d}|$. Let $S = \{k_{\lfloor\frac{d}{2}\rfloor+1},\dots, k_{d}\}$. Then $\varepsilon \geq \sum_{j=1}^{\lfloor d/2\rfloor} |x_{k_j}| \geq \frac{d}{2} |x_{i}|$ $\forall i\in S$, which implies $|x_i| \leq \frac{2\varepsilon}{d}$ $\forall i\in S$. 
\end{proof}

\begin{theorem}\label{thr:lower_bounds}
    Let $n,k,d\in\mathbb{N}$ and $\varepsilon>0$. With entry-wise query access to $V = [v_1,\dots,v_n]\in\mathbb{R}^{d\times n}$ and $(\|v_i\|)_{i\in[n]}$ and classical description of partition $\{\mathcal{C}_j\}_{j\in[k]}$ of $[n]$ with $|\mathcal{C}_j| = \Omega(\frac{n}{k})$, outputting centroids $c_1,\dots,c_k\in\mathbb{R}^d$ such that $\big\|c_j - \frac{1}{|\mathcal{C}_j|}\sum_{i\in\mathcal{C}_j} v_i \big\| \leq \varepsilon$ for all $j\in[k]$ has randomised query complexity $\Omega\big({\min}\big\{\frac{\|V\|_F^2}{n}\frac{kd}{\varepsilon^2}, nd\big\}\big)$ and quantum query complexity $\Omega\big({\min}\big\{\frac{\|V\|_F}{\sqrt{n}}\frac{kd}{\varepsilon}, nd\big\} \big)$.
\end{theorem}
\begin{proof}
    Let $\{\mathcal{C}_j\}_{j\in[k]}$ be such that $\mathcal{C}_j = \{(j-1)\frac{n}{k} + 1,(j-1)\frac{n}{k}+2,\dots,j\frac{n}{k}\}$ for all $j\in[k]$. Consider the initial centroids $c_1^0,\dots,c_k^0\in\mathbb{R}^d$ defining $\{\mathcal{C}_j\}_{j\in[k]}$ as $c^0_j = (\frac{j}{k},0,0,\dots,0)$, i.e., its first entry is $\frac{j}{k}$. Now let $\alpha\in\mathbb{R}_+$ be a positive number to be determined later and $W := \{w\in\{0,1\}^d:|w| = \lfloor\frac{d}{2}\rfloor\}$. Note that $\|w\| = \sqrt{\lfloor d/2\rfloor}$ for all $w\in W$. We claim that there is a rotation around the origin $R:\mathbb{R}^d\to\mathbb{R}^d$ such that $(Rw)_1 = 0$ for all $w\in W$, i.e., the first entry of $Rw$ is zero. Let then $V = [v_1,\dots,v_n]\in\mathbb{R}^{d\times n}$ be such that, for each $j\in[k]$, the vectors $\{v_i\}_{i\in \mathcal{C}_j}$ are $v_{i} = \alpha c_j^0 + \alpha R w_i$ (here the multiplication by $\alpha$ is done entry-wise), where $w_i$ is randomly picked from $W$. To be more precise, we pick the first $\lfloor\frac{d}{2}\rfloor$ bits of $w_i$ completely randomly, and the next $\lfloor\frac{d}{2}\rfloor$ bits as the complement of the first half (plus a final $0$ bit if $d$ is odd). This means that the vectors $\{v_i\}_{i\in \mathcal{C}_j}$ belong to the $(d-1)$-sphere of diameter $\Theta(\alpha\sqrt{d})$ centered at $\alpha c_j^0$ and on the hyperplane orthogonal to $c_j^0$. Moreover, by construction, $\|v_i\| = \alpha\sqrt{\frac{j^2}{k^2} + \lfloor\frac{d}{2}\rfloor}$ is constant for all $i\in\mathcal{C}_j$, so access to $\|v_i\|$ does not give any meaningful information about $c_j$. Now, notice that
    \begin{align*}
        \|V\|_F^2 &= \alpha^2\sum_{j\in[k]}\sum_{i\in\mathcal{C}_j} \|c_j^0 + R_jw_i\|^2 \leq \alpha^2\sum_{j\in[k]}\sum_{i\in\mathcal{C}_j} 2(\|c_j^0\|^2 + \|R_jw_i\|^2) \leq 2n \alpha^2(1 + d/2)\\
        \implies \alpha &\geq \frac{\|V\|_F}{\sqrt{2n}}\frac{1}{\sqrt{1+ d/2}} = \Omega\left(\frac{\|V\|_F}{\sqrt{n}}\frac{1}{\sqrt{d}} \right).
    \end{align*}

    Assume we have an algorithm that outputs $c_1,\dots,c_k\in\mathbb{R}^d$ such that $\big\|c_j - \frac{1}{|\mathcal{C}_j|}\sum_{i\in\mathcal{C}_j} v_i \big\| \leq \varepsilon$ for all $j\in[k]$. This allows us to output $\widetilde{w}_j := \frac{|\mathcal{C}_j|}{\alpha}R^{-1}(c_j - \alpha c_j^0) = \frac{n}{\alpha k}R^{-1}(c_j - \alpha c_j^0)$. Consider the first $\lfloor\frac{d}{2}\rfloor$ bits of $\widetilde{w}_j$ and $\sum_{i\in\mathcal{C}_j} w_i$ only. Then
    \begin{align*}
        \sum_{\ell=1}^{\lfloor d/2\rfloor}\Bigg|\widetilde{w}_{j\ell} - \sum_{i\in\mathcal{C}_j} w_{i\ell}\Bigg| \leq \sqrt{\left\lfloor\frac{d}{2}\right\rfloor}\Bigg\|\widetilde{w}_j - \sum_{i\in\mathcal{C}_j} w_i\Bigg\| = \sqrt{\left\lfloor\frac{d}{2}\right\rfloor}\Bigg\|R\widetilde{w}_j - \sum_{i\in\mathcal{C}_j} R w_i\Bigg\| \leq \frac{n\varepsilon}{\alpha k}\sqrt{\left\lfloor\frac{d}{2}\right\rfloor}
    \end{align*}
    for all $j\in[k]$.  According to \Cref{fact:simple_fact}, for each $j\in[k]$ there is $S_j\subseteq[\lfloor\frac{d}{2}\rfloor]$ with $|S_j| \geq \lfloor\frac{d}{4}\rfloor$ such that $|\widetilde{w}_{j\ell} - \sum_{i\in\mathcal{C}_j} w_{i\ell}| \leq \frac{4n\varepsilon}{\alpha k\sqrt{d}}$ for $\ell\in S_j$, i.e., the number $\widetilde{w}_{j\ell}$ approximates $\sum_{i\in\mathcal{C}_j} w_{i\ell}$ up to additive error $\frac{4n\varepsilon}{\alpha k\sqrt{d}}$ for all $\ell\in S_j$ and $j\in[k]$. This means that we can approximate the Hamming weight of $k\lfloor\frac{d}{4}\rfloor$ \emph{independent} bit-strings on $|\mathcal{C}_j| = \frac{n}{k}$ bits up to additive error $\frac{4n\varepsilon}{\alpha k\sqrt{d}}$ (the first $\lfloor\frac{d}{2}\rfloor$ bits of $w_i$ are independent by construction). According to \Cref{fact:lower_bound_reduction}, the randomized and quantum query lower bounds for approximating $k\lfloor\frac{d}{4}\rfloor$ independent Hamming weights on $\frac{n}{k}$ bits each to precision $\frac{4n\varepsilon}{\alpha k\sqrt{d}}$ are, respectively,
    \begin{align*}
        \Omega\bigg(kd\min\bigg\{\frac{n^2}{k^2}\frac{\alpha^2 k^2 d}{n^2 \varepsilon^2}, \frac{n}{k}\bigg\}\bigg) &= \Omega\bigg({\min}\bigg\{\frac{\|V\|_F^2}{n}\frac{kd}{\varepsilon^2}, nd\bigg\} \bigg),\\
        \Omega\bigg(kd\min\bigg\{\frac{n}{k}\frac{\alpha k\sqrt{d}}{n\varepsilon}, \frac{n}{k} \bigg\} \bigg) &= \Omega\bigg({\min}\bigg\{\frac{\|V\|_F}{\sqrt{n}}\frac{kd}{\varepsilon}, nd\bigg\} \bigg). \qedhere
    \end{align*}
\end{proof}

\paragraph{Acknowledgments.}
ET thanks Haotian Jiang for initial help with the analysis. JFD and AL thank Varun Narasimhachar for useful discussions. JFD thanks Miklos Santha for helpful discussions regarding the lower bounds. AC is supported by a
Simons-CIQC postdoctoral fellowship through NSF QLCI Grant No.\ 2016245. JFD and AL acknowledge support from the National Research Foundation, Singapore and A*STAR under its CQT Bridging Grant and its Quantum Engineering Programme under grant NRF2021-QEP2-02-P05. JFD is also supported by ERC grant No.\ 810115-DYNASNET. ET acknowledges support from NSF GRFP (DGE-1762114) and the Miller Institute for Basic Research in Science, University of California Berkeley. Part of this work was done while JFD and AL were visiting the Simons Institute for the Theory of Computing.

\bibliographystyle{plain}
\bibliography{biblio.bib}

\end{document}